\documentclass[preprint,11pt]{elsarticle}

\usepackage{lineno,hyperref}

\usepackage{array}
\usepackage{amsmath}
\usepackage{umoline}
\usepackage{amssymb}
\usepackage{graphicx}
\usepackage{multirow}
\usepackage{url}
\usepackage{amsfonts}
\usepackage{listings}
\usepackage{algorithm}
\usepackage[noend]{algpseudocode}
\usepackage{mathtools}
\usepackage{multirow}
\usepackage{url}
\usepackage{amsthm}
\usepackage{subfigure}
\newtheorem{lemma}{Lemma}
\newtheorem{theorem}{Theorem}

\renewcommand{\tabcolsep}{0.7mm}
\newcommand{\zipf}{\text{Zipf}}
\usepackage[margin=2.54cm]{geometry}
\usepackage{etoolbox}
\usepackage{xcolor}

\makeatletter
    \patchcmd{\tnotemark}{\ding{73}}{\dag}{}{\@latex@error{Failed to path \string\tnotemark\space for \string\ding{73}}}
    \patchcmd{\tnotemark}{\ding{73}\ding{73}}{\dag\dag}{}{\@latex@error{Failed to path \string\tnotemark\space for \string\ding{73}\string\ding{73}}}
    \patchcmd{\tnotetext}{\ding{73}}{\dag}{}{\@latex@error{Failed to path \string\tnotetext\space for \string\ding{73}}}
    \patchcmd{\tnotetext}{\ding{73}\ding{73}}{\dag\dag}{}{\@latex@error{Failed to path \string\tnotetext\space for \string\ding{73}\string\ding{73}}}
\makeatother

\journal{Journal of \LaTeX\ Templates}

\usepackage{mathtools}
\usepackage{multirow}
\usepackage{url}
\begin{document}
\begin{frontmatter}

\title{Compressed Key Sort and Fast Index Reconstruction}


\author[SNU,SAP]{Yongsik Kwon}
\author[SNU]{Cheol Ryu}
\author[SNU]{Sang Kyun Cha}
\author[NY]{Arthur H. Lee}
\author[SNU]{Kunsoo Park}
\author[SNU]{Bongki Moon}

\address[SNU]{Seoul National University, Seoul, Korea}
\address[SAP]{SAP Labs Korea, Seoul, Korea}
\address[NY]{State University of New York Korea, Incheon, Korea}

\begin{abstract}
In this paper we propose an index key compression scheme based on the notion of distinction bits by proving that the distinction bits of index keys are sufficient information to determine the sorted order of the index keys correctly.
While the actual compression ratio may vary depending on the characteristics of datasets (an average of 2.76 to one compression ratio was observed in our experiments), the index key compression scheme leads to significant performance improvements during the reconstruction of large-scale indexes. Our index key compression can be effectively used in database replication and index recovery of modern main-memory database systems.
\begin{keyword}
Compressed key sort \sep Distinction bit \sep Index reconstruction \sep Parallel sorting
\end{keyword}
\end{abstract}
\end{frontmatter}

\section{Introduction}

Main-memory database systems have been widely used for many applications
such as OLTP and OLAP, which are required to keep the latency low and
the transaction throughput high.
In such a main-memory database system, indexes are often deployed without
on-disk representations~\cite{Zhang16hstore,sqlserver,hekaton,MMDB}.
By letting the indexes reside solely in memory,
it can sustain the best attainable performance of the indexes,
which are already critical to query and data processing performance,
even in the presence of many updates.
Insertions, deletions and updates made to a database table will be reflected
to all the indexes associated with the table as well as the table itself.
However, none of those update operations will incur any disk accesses
for keeping the indexes up to date,
because all the corresponding changes will only be applied to the table
and the associated indexes residing in the main memory.
The changes applied to the indexes will not even be written to the log
or checkpointed to disk, as the most recent copy of each index can always be
restored from its base table~\cite{malviya14}.

Since none of the index updates are propagated to disk, however, all the indexes have to be reconstructed from scratch when the database system restarts from a failure, an anti-cached table is loaded back from disk to memory, or an entire database is replicated from the master node to a slave node. For a table that has many indexes associated with it, the cost of loading the table may be significantly increased due to the additional cost of constructing the indexes from the rows of the table. Therefore, it is practically an important challenge to limit the cost of
index reconstruction so that the database loading time and the restart time can be kept to its minimum.

\begin{table}
\renewcommand{\tabcolsep}{3mm}
\begin{center}
\begin{tabular}{c|c|ccc} 
\hline
&load&\multicolumn{3}{c}{index construction time} \\
table&time&sort&build&total\\
\hline
INDBTAB & 1.24 & 4.25 & 0.34 & 4.59 \\
Human & 5.34 & 19.94 & 1.17 & 21.12 \\ 
Wikititle & 1.36 & 5.09 & 0.59 & 5.68 \\
ExURL & 0.96 & 16.82 & 0.81 & 17.63 \\
WikiURL & 1.36 & 18.56 & 0.80 & 19.36 \\
Part & 0.25 & 0.65 & 0.05 & 0.70 \\
\hline
\end{tabular}
\caption{Index construction time for sample tables (in seconds).}
\label{tbl-indexcost}
\end{center}
\end{table}

Table~\ref{tbl-indexcost} shows the times taken to build a B-tree index for
each of the six memory-resident tables. The times are broken down to
two separate stages, namely, {\em sort} and {\em build}.
Both the sort and build phases of index construction were performed by
a single-core implementation.
Just for the sake of comparison, the second column of the table
shows the times taken to load the indexed columns from disk to memory.
In all the six cases reported in Table~\ref{tbl-indexcost},
the cost of internal sort was approximately 90 percent of the total cost.
Evidently, the internal sort was the dominant factor of
the B-tree index construction, and
hence fast construction of memory resident B-tree indexes cannot be
achieved without reducing the cost of internal sort drastically.

In this paper, we propose a new sort approach relying on the distinction bits
among the keys. We call this method a {\em compressed key sort},
as utilizing only the distinction bits of the keys can be considered
a kind of compressing the keys.
The experimental evaluation demonstrates that the overall cost of
B-tree construction can be reduced by 21--54 percent for real-world datasets.
Our experiments also show that the compressed key sort is readily
parallelizable for multi-core processors, yields near-linear speedup,
and can actually build a B-tree index faster than loading its image
from disk or even enterprise-class SSD.

The key question we pose in this paper is ``What is the minimum amount of
information required (in terms of the number of bits in the index keys)
to determine the sorted order among the index keys correctly?''
Whoever can determine it will extract the minimum number of bits from
index keys and sort them still correctly but more efficiently.
Once the keys are sorted, the B-tree index will be built by
following the standard bulk-load procedure.
This process is illustrated in Figure~\ref{fig-sortkey},
where the top flow shows the conventional steps for index construction
while the bottom one shows the proposed compressed key sort applied to
index construction.

We now formally define the {\em distinction bits} of two index keys to be the most significant bits that are different between the keys. We will prove that the distinction bits of index keys are sufficient information to determine the sorted order of the index keys correctly. Consequently, we can extract only the distinction bits of index keys into compressed form and sort the compressed keys in order to construct a B-tree
index quickly.

Index keys for database tables can be as short as a 4-byte word but they can also be longer than a few dozen bytes in business applications. Hence index trees and all related algorithms (sorting index keys, building the index, searching with a query key, etc.) should be able to handle long keys as well as short keys. Our compressed key sort approach assumes this wide range of index key sizes. 
To speed up the index construction process further, we exploit parallelism in building an index tree. 

\begin{figure}
\centering
  \includegraphics[height=6.8cm]{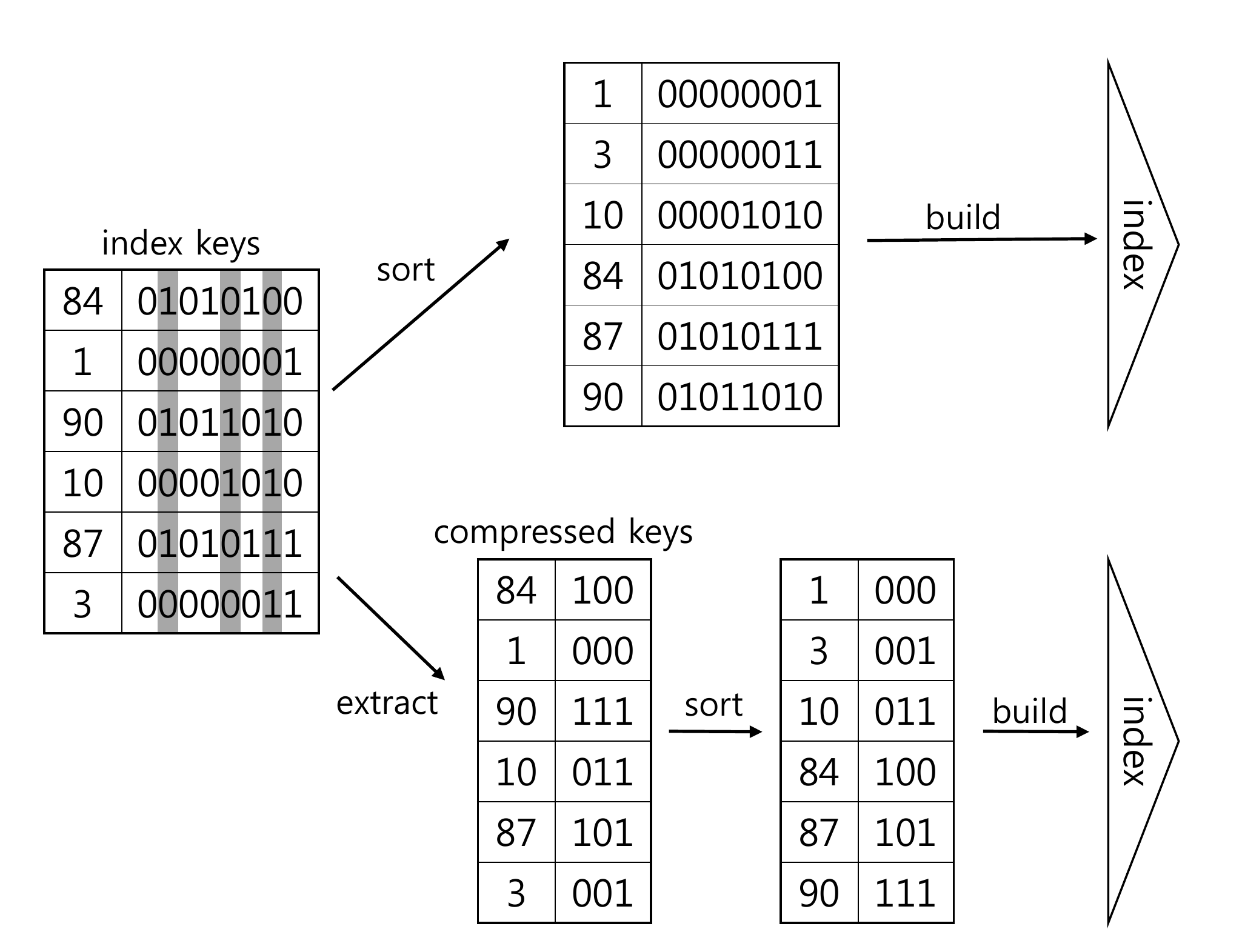}
\caption{Compressed key sort.}
\label{fig-sortkey} 
\end{figure}

This paper is organized as follows. In Section 2 we discuss related work. In Section 3 we introduce our compressed key sort. In Section 4 we describe the index key formats we use for various data types and present the metadata information to keep for efficient index rebuilding. In Section 5 we explain the procedure of rebuilding the index. Section 6 shows the results of our experiments, and we conclude in Section 7.

\section{Related Work}
There has been extensive research on efficient index structures for database tables, where efficiency measures are index size, search time, concurrency control, etc. Especially the following work focused on reducing index sizes and/or search time: Bayer and Unterauer's Prefix B-tree \cite{BU}, Lehman and Carey's T-tree \cite{LC}, Ferguson's Bit-tree \cite{Fe}, Bohannon et al.'s partial-key T-tree and partial-key B-tree \cite{BMR}, Rao and Ross's CSS-tree \cite{RR99} and CSB+ tree \cite{RR00}, Chen et al.'s pB+ tree \cite{Chen01} and fpB+ tree \cite{Chen02}, Schlegel et al.'s k-ary search tree \cite{KARY}, Boehm et al.'s Generalized Prefix Tree \cite{BSV}, Kissinger et al.'s KISS-tree \cite{KISS}, and more recently Kim et al.'s Fast Architecture-Sensitive Tree (FAST) \cite{KCS}, Yamamuro et al.'s VAST-tree \cite{YOH}, Levandoski et al.'s Bw-tree \cite{bwtree}, Leis et al.'s Adaptive Radix Tree (ART) \cite{LKN}, Zhang et al.'s SuRF \cite{SuRF},
Binna et al.'s HOT \cite{HOT}.
See Graefe and Larson \cite{GL} and Graefe \cite{G11} for surveys.

Our work reduces the sizes of sort keys by compressing them, from which the speedups in sorting and index rebuilding are obtained. Hence our work is orthogonal to the previous work on efficient index structures, and it can be applied to many index structures. Compressing keys by distinction bits can also be applied to big data file formats. Popular self-described file formats such as ORC \cite{ORC} and Parquet \cite{PAR} adopt columnar storage structures to cope with read-heavy analytic workloads against large-scale distributed datasets. Compression by distinction bits can accelerate such common analytic tasks as sorting data and generating unique keys.

Kim et al.'s FAST \cite{KCS} proposed a key compression technique which extracts bits of index keys in the bit positions where the index keys are not the same (which are called \emph{variant bits}). 
We go one step further and use distinction bits to determine the sorted order of index keys correctly.

There has been research on \emph{order-preserving compression} \cite{OP,DOP,QUERY,DOPC,OPKC} which
maps index keys into encoded values such that the order of index keys is the same as the order of encoded values. In order-preserving compression, index keys are replaced by encoded values. In our compression scheme, however, there is no encoding. We simply extract part of index keys (i.e., distinction bits) to speed up sorting and index building. The index tree built by our compression scheme will be a conventional B-tree index without any encoding of index keys.

For sorting in multi-core CPUs, there have recently been many results exploiting SIMD parallelism \cite{IMK,CMB,SKC}.
In our target applications which require a wide range of index key sizes, however, the size of index keys is too big to exploit SIMD parallelism. Thus we implemented our own parallel sorting algorithm called the {\em row-column sort}, which relies only on the operation of comparing two elements during sorting. The comparison operator is called a \emph{comparator}. Therefore, the row-column sort works with any key sizes. (In contrast, sorting on SIMD needs quite different algorithmic techniques such as merging networks \cite{IMK,CMB,SKC}.)
In our experiments we compared the row-column sort with GCC STL parallel sort \cite{GCC},
which is an available parallel sort code on multi-core CPUs in which a custom comparator can be used.
Experiments show that the row-column sort shows a better speedup than GCC STL sort, and
it is 31.4\% faster than GCC STL sort when the number of cores is 16.

\section{Compressed Key Sort}
In this section we first prove that the distinction bits of index keys are sufficient information to determine the sorted order of the index keys correctly, and then
present our compressed key sort based on distinction bits, which is the central idea in rebuilding main-memory indexes efficiently. 

\begin{figure}
\centering
\includegraphics[height=8cm]{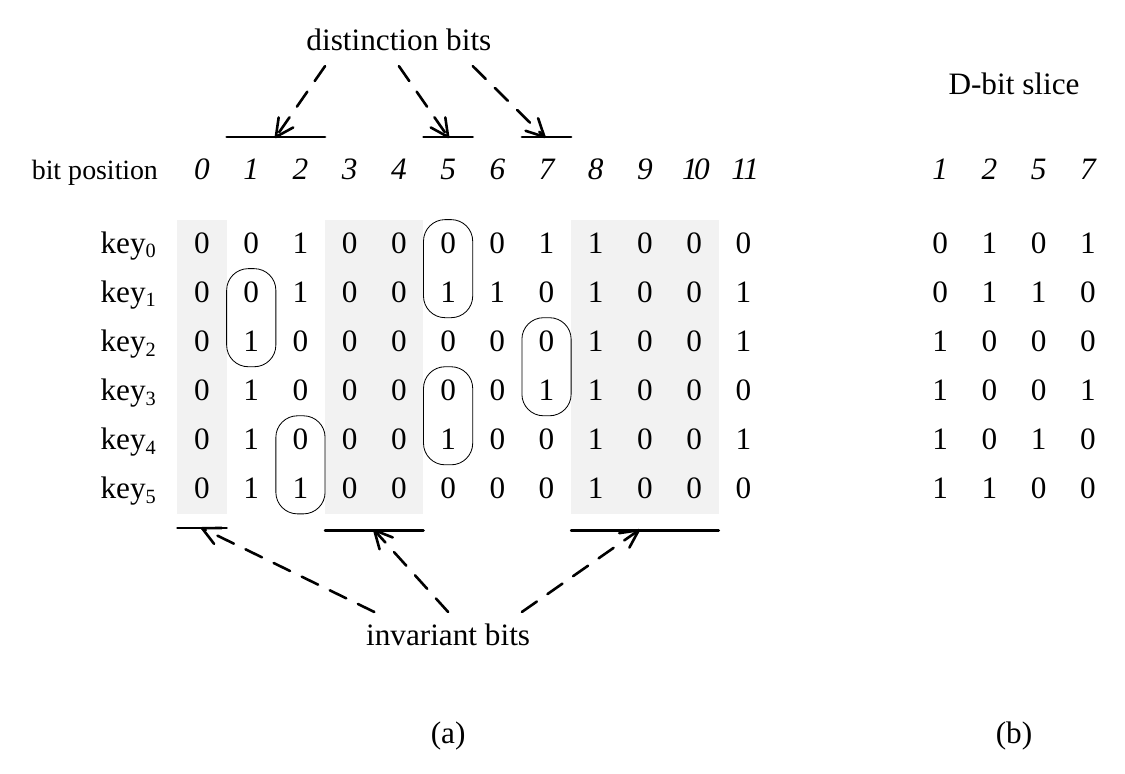}
\caption{Distinction bits and invariant bits.}
\label{fig-dbit}
\end{figure}

\subsection{Definitions}
We first introduce some terms to describe the compressed key sort. We consider key values as binary strings throughout this paper. The bit positions where all key values of the given dataset are identical are called \emph{invariant bit positions} (the bits in these positions are called \emph{invariant bits}). The other bit positions are called \emph{variant bit positions} (the bits themselves are called \emph{variant bits}). Each row in Figure~\ref{fig-dbit} (a) represents a key value. In the figure, bit positions 0, 3, 4, 8, 9 and 10 are invariant bit positions, and bit positions 1, 2, 5, 6, 7 and 11 are variant bit positions. Note that bit positions start with 0, and the bit of a key in bit position $i$ will be called the $(i+1)$-st bit of the key (i.e., the bit in bit position 0 is the first bit, the bit in bit position 1 is the second bit, etc.).
The first bit (i.e., in bit position 0) is the most significant bit in the keys.

The \emph{distinction bit position} of two keys is defined as the most significant bit position where the two keys differ (the bits themselves are called \emph{distinction bits}). The name \emph{distinction bit}\footnote{The name \emph{discriminative bit} is used in \cite{HOT}.} is from \cite{Fe}, but the main focus of this paper is sorting based on distinction bits\footnote{This idea was also described in \cite{Patent}.}.

Suppose that there are $n+1$ keys $key_0,key_1,\ldots, key_n$ and they are in lexicographic order (sorted order), i.e., $key_0 < key_1 < \cdots < key_n$. 
The distinction bit position of two keys $key_i$ and $key_j$ is denoted by $\text{D-bit}(key_i, key_j)$.
Let $D_i = \text{D-bit}(key_{i-1}, key_i)$ for $1\leq i\leq n$, i.e., the distinction bit position of two adjacent keys in sorted order. 
We prove that the set of distinction bit positions of all possible key pairs is the same as the set $\{D_1,D_2,\ldots,D_n\}$. First, we need a lemma.

\begin{lemma}\label{lemma-dbit} 
$\text{D-bit}(key_i, key_j) = \min_{i < k \leq j} D_k$ for all $0\leq i < j\leq n$.
\end{lemma}

\begin{proof}
We prove by induction on $d = j-i$. When $d=1$, the lemma holds trivially. For induction hypothesis, assume that the lemma holds for $d\geq 1$. 

We now prove the lemma for $d+1$. Let $D = \text{D-bit}(key_i, key_{j-1})$. Because $(j-1)-i = d$, $D=\min_{i < k\leq j-1}D_k$ by induction hypothesis, which means that the first $D$ bits of $key_i, key_{i+1},\ldots, key_{j-1}$ are the same. Consider $D$ and $D_j$ ($=\text{D-bit}(key_{j-1}, key_j)$). 
\begin{itemize}
\item If $D<D_j$, then $D$ is D-bit$(key_i, key_j)$ because
the $(D+1)$-st bit of $key_j$ is the same as that of $key_{j-1}$
which is different from that of $key_i$, while the first $D$ bits of $key_i, key_{i+1},\ldots, key_j$ are the same.
Since $D=\min_{i < k\leq j-1}D_k$ and $D<D_j$, $D = \min_{i < k \leq j} D_k$.
\item If $D_j<D$, then $D_j$ is $\text{D-bit}(key_i, key_j)$ because
the $(D_j+1)$-st bit of $key_j$ is different from that of $key_{j-1}$ which is the same as that of $key_i$, while the first $D_j$ bits of $key_i, key_{i+1},\ldots, key_j$ are the same.
Since $D=\min_{i < k\leq j-1}D_k$ and $D_j<D$, $D_j = \min_{i < k \leq j} D_k$.
\end{itemize} 
For example, $\text{D-bit}(key_1, key_3) = 1$ because $D_2 = 1 < D_3 = 7$,
and $\text{D-bit}(key_0, key_2) = 1$ because $D_1 = 5 > D_2 = 1$ in Figure~\ref{fig-dbit} (a).
Note that $D$ cannot be equal to $D_j$ because we have only two possibilities, 0 and 1, in a bit position. 
\end{proof}

\begin{theorem}\label{theorem-dbit} 
The set $D_{all}$ of distinction bit positions of all possible key pairs is the same as the set $D_{adj}=\{D_1,D_2,\ldots,D_n\}$, i.e., the set of distinction bit positions of adjacent keys in sorted order.
\end{theorem}

\begin{proof}
Since adjacent key pairs are part of all key pairs, we have $D_{adj}\subseteq D_{all}$.

To prove $D_{all}\subseteq D_{adj}$, we show that the distinction bit position of any pair (say, $key_i$ and $key_j$) belongs to $D_{adj}$.
Without loss of generality, assume that $i<j$.
By Lemma~\ref{lemma-dbit}, $\text{D-bit}(key_i, key_j) = \min_{i < k \leq j} D_k$.
Since $D_k$ is in $D_{adj}$, so is $\text{D-bit}(key_i, key_j)$.
\end{proof}

By Theorem~\ref{theorem-dbit}, all possible distinction bit positions for $n+1$ keys are $D_1,\ldots,D_n$ (i.e., there are at most $n$ distinction bit positions),
which is a crucial fact in our compressed key sort.
Since some of $D_i$ values may be the same (e.g., $D_1=D_4=5$ in Figure~\ref{fig-dbit}),
the number of distinction bit positions can be much less than $n$.
In Figure~\ref{fig-dbit} (a), bit positions 1, 2, 5 and 7 are distinction bit positions, because $D_1 = 5$, $D_2 = 1$, etc. It is obvious that distinction bit positions are variant bit positions. However, there may be variant bit positions which are not distinction bit positions. In Figure~\ref{fig-dbit} (a), bit positions 6 and 11 are such positions.

\subsection{Key Compression}
\emph{Extended distinction bit positions} mean all distinction bit positions plus some other (zero or more) bit positions.
Let $\text{Compress}(key_i)$ be the concatenation of the bits of $key_i$ in extended distinction bit positions. The \emph{distinction bit slice} (or \emph{D-bit slice}) is defined as the set  $\{\text{Compress}(key_0),\ldots,\break\text{Compress}(key_n)\}$.
See Figure~\ref{fig-dbit} (b). The distinction bit slice is simply a set of $\text{Compress}(key_i)$'s, not necessarily sorted by $\text{Compress}(key_i)$.

\begin{theorem}\label{theorem-dbitslice}
The distinction bit slice is sufficient information to determine the sorted order of the keys.
\end{theorem}

\begin{proof}
We first prove the theorem for the distinction bit positions. 
We prove that the following relation holds: 
\[
key_i < key_j\mbox{ if Compress}(key_i) < \mbox{Compress}(key_j)
\]
for all $i$ and $j$. Let $D = \text{D-bit}(key_i, key_j)$. Since the first $D$ bits of $key_i$ and $key_j$ are the same, the order of $key_i$ and $key_j$ is determined by the bits in bit position $D$. By Lemma~\ref{lemma-dbit}, bits in bit position $D$ are in Compress (and thus in the distinction bit slice). Hence, the order between $key_i$ and $key_j$ is determined by the order between $\text{Compress}(key_i)$ and $\text{Compress}(key_j)$.

Due to the relation above, we can correctly determine the sorted order of the keys by the distinction bit slice. Similarly, we can prove the theorem for extended distinction bit positions.
\end{proof}



When we maintain an index for a database table, index keys may be inserted, deleted, or updated by database operations. Then distinction bit positions may be changed at runtime. For example, if $key_3$ is deleted in Figure~\ref{fig-dbit}, position 7 is no longer a distinction bit position (but it is still a variant bit position). If $key_0$ is also deleted, distinction bit positions don't change, but position 7 becomes an invariant bit position. If an index key is inserted, a new distinction bit position may be added. 
It is quite expensive to maintain the distinction bit positions accurately when delete operations are allowed. 
Theorem~\ref{theorem-dbitslice} allows us to lazily update  distinction bit positions without affecting the correctness of sorting by letting some invalidated bit positions stay.

The scheme of extracting distinction bits of $key_i$ into $\text{Compress}(key_i)$ for all $i$ and sorting $\text{Compress}(key_i)$'s 
rather than sorting full key values is called \emph{compressed key sort}.
In order to extract compressed keys from index keys,
we need to keep only (extended) distinction bit positions as a bitmap.
Compressed key sort is the main reason for the speedup of index reconstruction.

\begin{figure}
\centering
\includegraphics[height=3cm]{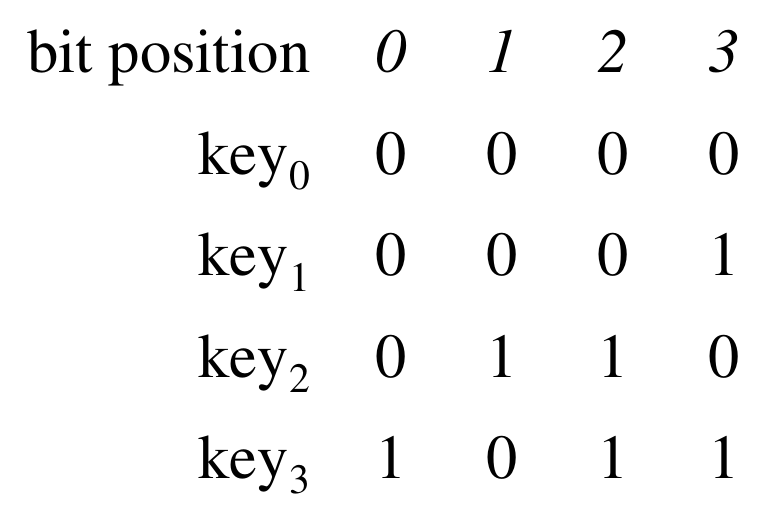}
\caption{Distinction bit positions and minimum bit positions.}
\label{fig-dbitpos}
\end{figure}

\vspace{6pt}
\noindent\textbf{Remark 1.}
To build an index, the sorting of index keys is necessary, which requires $O(n\log n)$ time. To compute distinction bit positions additionally, our compressed key sort needs $O(n)$ time to compare adjacent keys in sorted order.
However, our key compression is not optimal in terms of the number of bit positions if an unlimited time is allowed. For the given keys in Figure~\ref{fig-dbitpos}, our key compression selects bit positions 0, 1, and 3 as distinction bit positions, but the bit slice in bit positions 2 and 3 can correctly determine the sorted order of the keys, and this is the minimum number of bit positions.
An optimal algorithm can find the minimum number of bit positions by choosing every subset of the bit positions and checking whether the bit slice in the subset of bit positions can correctly determine the sorted order of the keys.
We conjecture that our key compression is best (in terms of the number of bit positions) if the sorting complexity (i.e., $O(n\log n)$ time) is allowed.

\begin{figure}
\centering
\includegraphics[height=4.5cm]{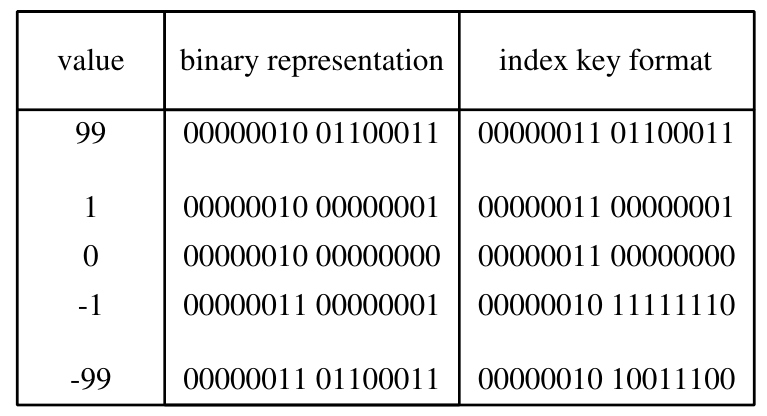}
\caption{Binary representations and index key formats of 2-byte decimal(2,0).}
\label{fig-binary}
\end{figure}

\section{Data Structures}
\subsection{Index Key Format}
The B+ tree and its variants are widely used as indexes in modern DBMSs to enable fast access to data with a search key. If an index is defined on columns $A_1, \ldots, A_k$ of a table, its key can be represented as a tuple of column values, of the form $(a_1, \ldots, a_k)$ \cite{SKS}. The ordering of the tuples is the lexicographic ordering.  When $k=2$, for example, the order of two tuples $(a_1,a_2)$ and $(b_1,b_2)$ is determined as follows: $(a_1,a_2) < (b_1,b_2)$ if $a_1 < b_1$ or ($a_1 = b_1$ and $a_2 < b_2$). 

In this section we describe how to make actual index keys from the tuples of column values so as to keep the lexicographic ordering of the tuples. We first explain how to make index keys from different data types and then explain how to make an index key from multiple columns.

For each data type (\texttt{int}, \texttt{float}, \texttt{decimal}, \texttt{string}, etc.), its index key format can be defined so that a lexicographic binary comparison in the index key format is equivalent to a comparison of original data values.
For the mappings of data types \texttt{int} and \texttt{float} to index key formats, we refer readers to \cite{LKN}. Here we describe the mappings of \texttt{decimal} and \texttt{string} to index key formats.

\begin{enumerate}



\item[A.] decimal: A decimal number $x$ is represented by a 1-byte header and a decimal part. The last bit of the header is the sign of the decimal number (1 for negative), and the second-to-last bit indicates whether the entry is null or not (0 for null). The decimal part contains a binary number corresponding to $x$ in $\lceil \log_2 (x+1) / 8\rceil$ bytes. 
The location of the decimal point is stored in the metadata of the column. For mapping, if the sign bit is 1 (i.e., $x$ is negative), toggle the sign bit and all bits of the decimal part; otherwise, toggle the sign bit only.
Then the order of the mapped values corresponds to that of the decimal numbers.
See Figure~\ref{fig-binary}, where decimal$(m,n)$ means $m$ total digits, of which $n$ digits are to the right of the decimal point.

\item[B.] fixed-size string: We use a fixed-size string as it is.

\item[C.] variable-size string with maximum length: 
A variable-size string with maximum length $n$ is denoted by varchar$(n)$.
We assume that the null character (denoted by $\emptyset$) is not allowed in the variable-size string. (In the case that null characters are allowed, we need to use some encoding of characters so that the encoded string doesn't have null characters.)
We attach one null character at the end of the variable-size string to make the index key value. 
Then the lexicographic order of index key values corresponds to that of variable-size strings as follows. If two index keys have the same lengths, the order between them is trivially the order of the strings. If two index keys have different lengths (let $k$ be the length of the shorter key) and their first $k-1$ bytes have different values, their order is determined by the first $k-1$ bytes. If two keys have different lengths and their first $k-1$ bytes have the same values, the shorter one is smaller in lexicographic order because it has a null character in the $k$-th byte and the longer one has a non-null character in the $k$-th byte. For instance, if two keys are AB$\emptyset$ and ABA$\emptyset$, then AB$\emptyset$ is smaller than ABA$\emptyset$ due to the 3rd bytes and this is the lexicographic order between two strings AB and ABA. Furthermore, the distinction bit position takes place in the null character of the shorter key.
\end{enumerate}
In each data type, the order between two index keys can be determined by a lexicographic binary comparison of them. 

\begin{figure}
\centering
\includegraphics[height=6cm]{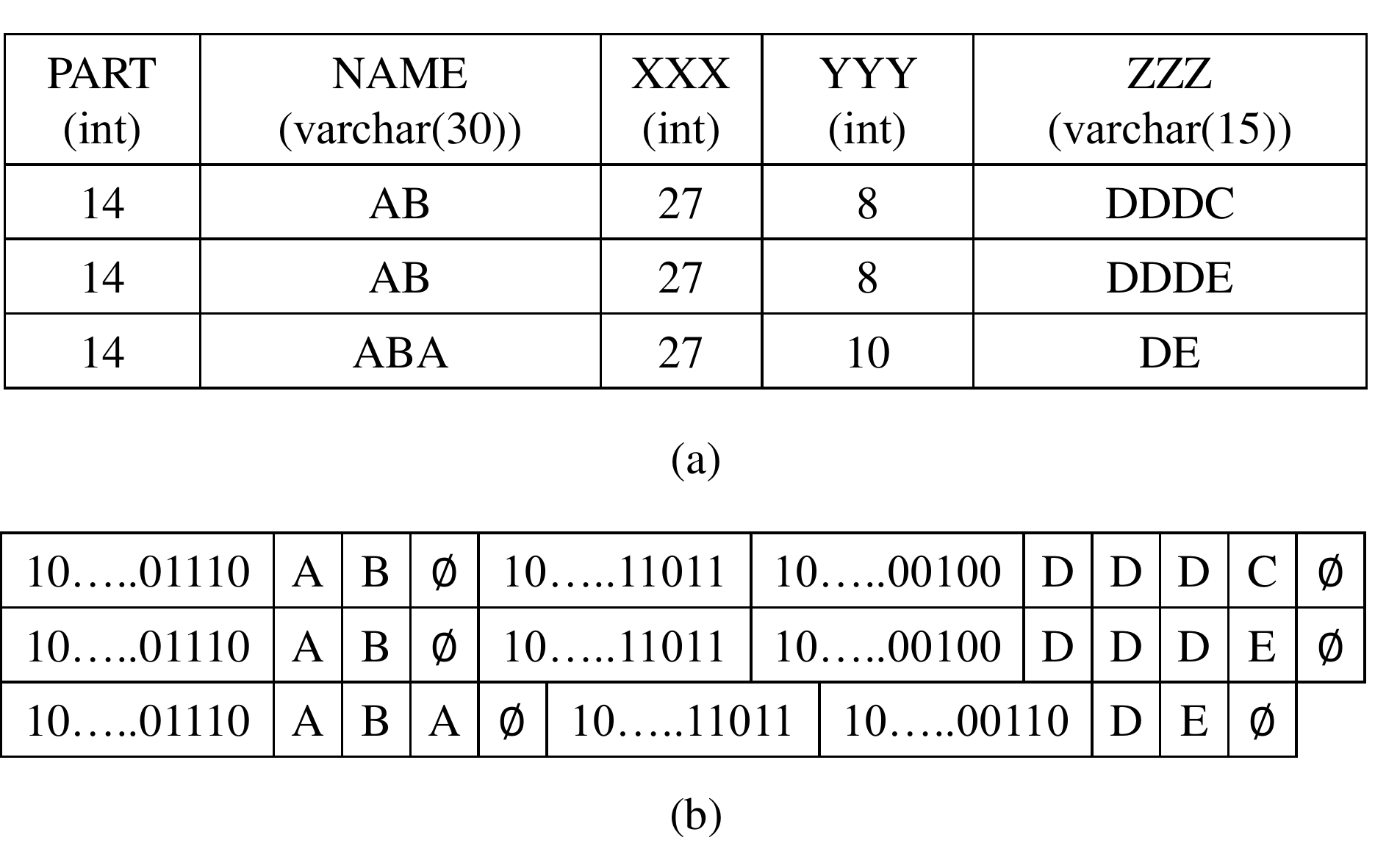}
\caption{Index keys from multiple columns. (a) Database table. (b) Index keys from the five columns.}
\label{fig-multicolumn}
\end{figure}

We now explain how to make an index key from multiple columns. An index key on multiple columns is defined as the concatenations of index keys from the multiple columns. Suppose that an index key is defined on the following five columns: PART (int), NAME (varchar(30)), XXX (int), YYY (int), and ZZZ (varchar(15)). Example column values in some rows are shown in Figure~\ref{fig-multicolumn} (a), and the index keys of the three rows are in Figure~\ref{fig-multicolumn} (b).

The distinction bit positions in Section 3.1 are defined on these full index keys. If the data types of index columns have fixed lengths (as in int, float, decimal, and fixed-size string), the column values are aligned in the index keys, and the order between index keys are determined by the lexicographic order of the column values.

If the data types of index columns have variable lengths (as in variable-size string), however, the column values may not be aligned in the index keys, as shown in Figure~\ref{fig-multicolumn} (b). Still we define distinction bit positions on these full index keys. If two rows have variable-size strings of different lengths in a column (e.g., column NAME in Figure~\ref{fig-multicolumn}), the distinction bit position takes place in that column as described above if previous columns have the same values as in Figure~\ref{fig-multicolumn}, and the order between the two index keys are determined by the lexicographic order of the variable-size strings in that column.

To compare two index keys, we make a binary comparison (by word sizes) of the two keys. If one index key is shorter, it is padded with 0's in the binary comparison. (The padded value does not affect the order of the two keys.) In this way we define distinction bits and distinction bit positions on full index keys derived from multiple columns.

\subsection{Index Tree and DS-metadata}
Although our compressed key sort can work with any variant of the B+ tree index structure, available codes for indexes have a small and fixed length for index keys (4 bytes for FAST \cite{FAST} and CSB+ tree \cite{CSB}, and 4 or 8 bytes for k-ary search tree \cite{KARY}) or some restrictions in building indexes (e.g., no parallel index building for ART \cite{LKN}). Therefore, we use a full-fledged index tree which is being used in SAP HANA database system, and 
apply our compressed key sort to it.
Figure~\ref{fig-indextree} shows the structure of the index tree, 
which is a variant of the partial-key B+ tree \cite{BMR}.
To define partial keys on key values $key_0,key_1,\ldots,key_n$ in sorted order, a parameter $pk$ is introduced.
The \emph{partial key} of $key_i$ is the $pk$ bits following the 
distinction bit position $D_i$ \cite{BMR}. In Figure~\ref{fig-dbit}, the partial key of $key_1$ when $pk=4$
is 1010, because $D_1=5$. The distinction bit position $D_i$ is also called the \emph{offset} of the partial key of $key_i$ \cite{BMR}.

\begin{figure}
\centering
\includegraphics[height=9cm]{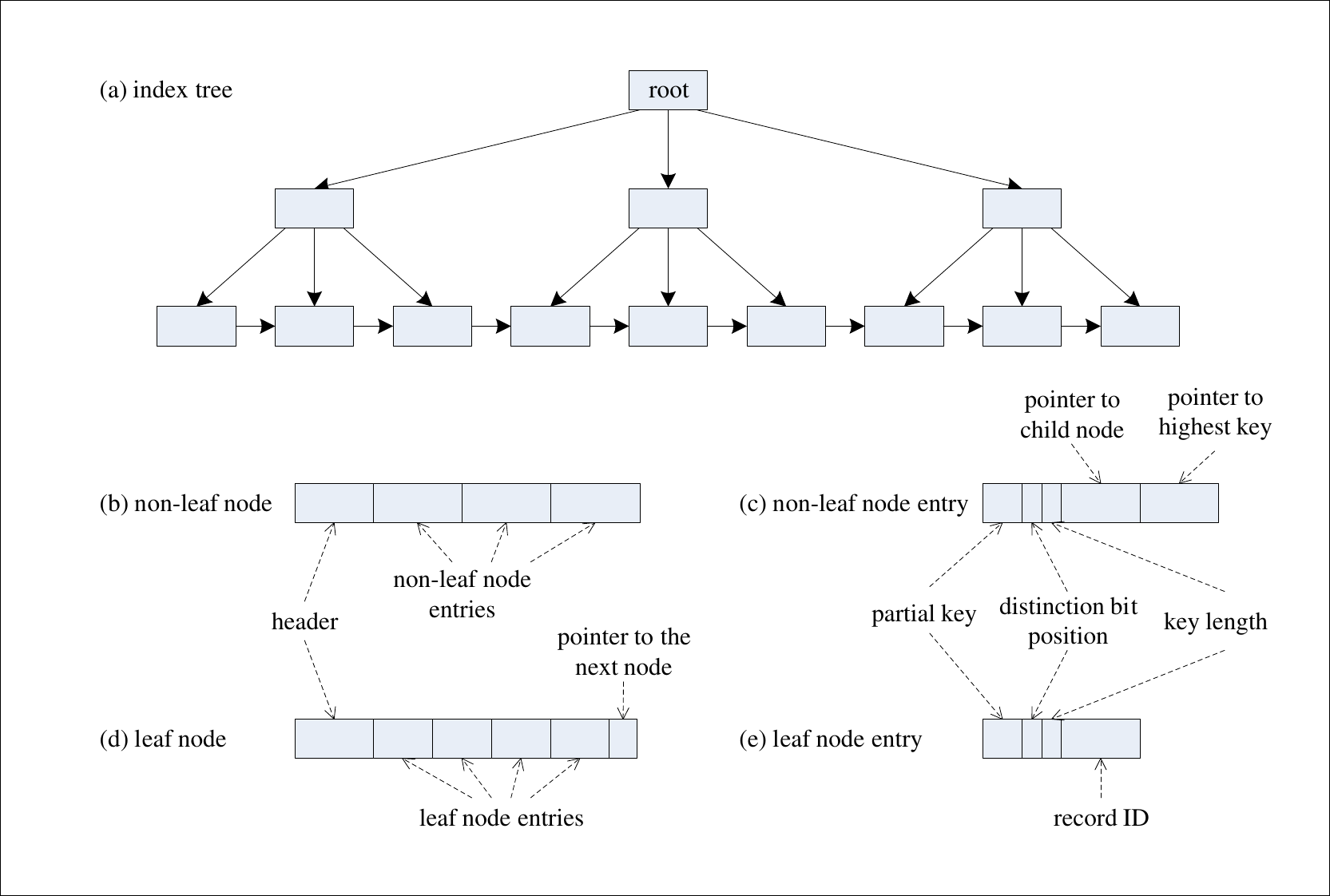}
\caption{Index tree structure.}
\label{fig-indextree}
\end{figure}


We describe the structure of the index tree in Figure~\ref{fig-indextree} which is relevant to this paper.
A leaf node of the index tree contains a list of entries, one for each index key,
plus a pointer to the next node. 
An entry in a leaf node consists of a partial key, a distinction bit position, an index key length, and a record ID. 
The header of a leaf node contains a pointer to the last (i.e., highest) index key of the entries in the leaf node.
A non-leaf node contains a list of entries, one for each child. 
The index key corresponding to an entry is the highest index key in the descendant leaves of the child corresponding to the entry.
An entry in a non-leaf node consists of a partial key, a distinction bit position, an index key length, a pointer to the child node, and a pointer to the highest index key (where the partial key and the index key length are those of the highest index key, and the distinction bit position is that of the highest index key against the highest index key of the previous entry).

In addition, we keep the following information \emph{persistently} for each index tree, which will be called the DS-metadata (DS stands for D-bit Slice).

\begin{enumerate}
\item D-bitmap: Our compression scheme requires distinction bit positions, which can be represented by a bitmap. The position of each bit in the bitmap means the position in the full index key. While the value 0 means that the bit position is not a distinction bit position, the value 1 means that it is possibly a distinction bit position.

\item Variant bitmap: Similarly we store variant bit positions in a bitmap, where value 0 in a bit position means that the bit position is not a variant bit position and value 1 means that it is possibly a variant bit position.

\item Reference key value: We need a reference key value for invariant bits, which can be an arbitrary index key value because the invariant bits are the same for all index keys.
\end{enumerate}

\noindent
Note that we use extended distinction bit positions to define the D-bitmap.
We maintain the variant bitmap and a reference key value in order to obtain partial keys when we rebuild our index tree. If partial keys are not needed in an index, the variant bitmap and a reference key value are not necessary, and we need only maintain the D-bitmap, which is the main information to keep for efficient index rebuilding.

\subsection{Search and Update Operations}

We describe how to perform search/insert/delete operations with the index tree and DS-metadata.
\begin{itemize}
\item {Search: Given a search key value $K$, we search down the index tree for $K$ as follows.

In a non-leaf node, we need to compare $K$ with an index key in a non-leaf node entry. Since the entry has a pointer to the highest index key (say, $A$), we make a binary comparison of two full key values $K$ and $A$. 

A leaf node contains a list of partial keys, and thus we need to compare search key $K$ with a list of partial keys. The procedure to compare $K$ with a list of partial keys is the same as the one described in Bohannon et al.~\cite{BMR}.}

\item {Insert: Given an insert key value $K$, insert $K$ into the index tree as follows.
\begin{enumerate}
\item {Search down the index tree with $K$ and find the right place for insertion (say, between two keys $A$ and $B$).}
\item {Compute the distinction bit positions $\text{D-bit}(A,K)$ and $\text{D-bit}(K,B)$.}
\item {Make changes in the index tree corresponding to the insertion, and update the D-bitmap and the variant bitmap as follows. For the D-bitmap, in principle we need to remove the bit position $\text{D-bit}(A,B)$ and add new distinction bit positions $\text{D-bit}(A,K)$ and $\text{D-bit}(K,B)$
because key $K$ has been inserted between $A$ and $B$. By Lemma~\ref{lemma-dbit}, however, $\text{D-bit}(A,B) = \min(\text{D-bit}(A, K), \text{D-bit}(K,B))$. Since the minimum position is $\text{D-bit}(A,B)$ and it is already set in the D-bitmap, we need only set $\max(\text{D-bit}(A, K), \text{D-bit}(K,B))$ in the D-bitmap if it is not already set. For the variant bitmap, we perform a bitwise XOR on $K$ and the reference key value, and then perform a bitwise OR on the variant bitmap and the result of the above bitwise XOR. 
The result of the bitwise OR will be the new variant bitmap. (Notice that the number of actual write operations on the D-bitmap is bounded by the number of 1's in the D-bitmap. Thus the chances that an actual write operation on the D-bitmap occurs during an insertion are low.)
}
\end{enumerate}
}

\item {Delete: Given a delete key value $K$, delete $K$ from the index tree as follows.

We delete $K$ as a usual deletion is done in the index tree, and simply leave the D-bitmap and the variant bitmap without changes. We need to show that the D-bitmap is valid after deleting $K$. Let $A$ and $B$ be the previous key value and the next key value of $K$, respectively. After deleting $K$, $\text{D-bit}(A,B)$ should be set in the D-bitmap. Again by Lemma~\ref{lemma-dbit}, $\text{D-bit}(A,B) = \min(\text{D-bit}(A,K), \text{D-bit}(K,B))$. Since $\text{D-bit}(A,K)$ and $\text{D-bit}(K,B)$ are set in the D-bitmap, $\text{D-bit}(A,B)$ is already set, whether it is $\text{D-bit}(A,K)$ or $\text{D-bit}(K,B)$.}
\end{itemize}
An update operation is done by a delete operation followed by an insert operation. Note that if there are only insert operations (i.e., no delete operations), the D-bitmap represents the distinction bit positions exactly. 

As the data in a database table change, the DS-metadata is updated incrementally as above. When an insertion occurs, at most one distinction bit position is added to the D-bitmap, and some variant bit positions may be added to the variant bitmap. This operation never reverts even if there is a delete or rollback, because implementing the revert exactly is quite expensive. Hence there may be positions in the D-bitmap whose values are 1, but which are not distinction bit positions. Also the variant bitmap may have positions whose values are 1, but which are not variant bit positions. However, they do not affect the correctness as shown in Theorem~\ref{theorem-dbitslice}. Such bit positions can be removed by scanning the index and computing the DS-metadata occasionally. If we rebuild the index anew, then certainly there will be no such bit positions.

With the current DS-metadata, we can rebuild the index tree (which will be described in the next section) when it is lost or unavailable. 
Even after the index tree is rebuilt, we may use the current DS-metadata as the DS-metadata. 
However, index rebuilding is an opportune time to compute the DS-metadata anew.
We compute the new DS-metadata from the current DS-metadata as follows.
\begin{enumerate}
\item D-bitmap: Extract compressed keys from index keys by the current D-bitmap, sort the compressed keys, and compute the distinction bit positions between adjacent compressed keys (all three steps are part of index reconstruction), which make the new D-bitmap.
Note that the bit positions where the current D-bitmap had 0 remain 0 in the new D-bitmap.
\item Reference key value: Take an arbitrary index key as the reference key value.
\item Variant bitmap: Initially the variant bitmap is all 0, and we take index keys one by one (say, $K$) and do the following. Perform a bitwise XOR on $K$ and the reference key value, followed by a bitwise OR on the variant bitmap and the result of the bitwise XOR (as in the insert operation above). 
\end{enumerate}
If we build an index tree for the first time (i.e., there is no DS-metadata at all), then we compute the D-bitmap as above, but with full index keys rather than compressed keys.

\vspace{6pt}
\noindent\textbf{Remark 2.}
Our key compression per se requires $O(n)$ time to compute the DS-metadata initially (other than sorting) and $O(1)$ time to update the DS-metadata for an insertion (other than $O(\log n)$ search time to find the right place to insert).
Note that sorting is needed anyway to build an index tree and a search is needed anyway to find the place to insert.
For the optimal algorithm described in Remark 1, finding the minimum number of bit positions after an insertion is very expensive.
Any \emph{practical} compression scheme should have low complexities in computing compression information such as the DS-metadata and updating the information.

\begin{figure}
\centering
\includegraphics[height=5.8cm]{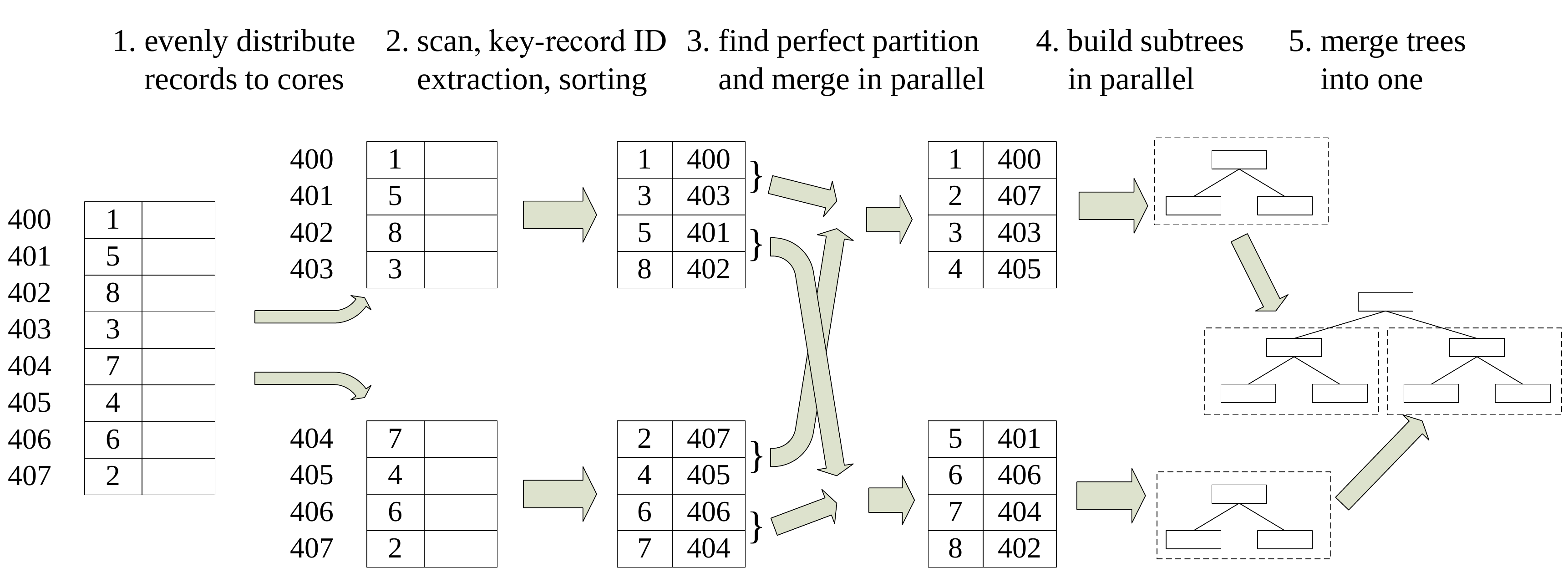}
\caption{Index-reconstruction procedure.}
\label{fig-indexrebuild}
\end{figure}

\section{Index Reconstruction}
We now describe how to build the index tree in parallel from a database table loaded in memory by using the DS-metadata on the fly. We extract only the bits from the index key values whose positions are set in the D-bitmap.
Figure~\ref{fig-indexrebuild} shows the overall procedure of parallel index reconstruction.

\begin{enumerate}
\item{To collect index keys in parallel, data pages of the target table are evenly distributed to the cores.}

\item{Each core scans the assigned data pages and extracts compressed keys and corresponding record IDs. A pair of a compressed key and the corresponding record ID makes a sort key. The record ID is included in the sort key so that each pair of a compressed key and a record ID can be directly used to fill its corresponding leaf node entry without causing many cache misses.}

\item{Sort the pairs of compressed key and record ID by a parallel sorting algorithm.}

\item{Build the index tree in a bottom-up fashion.}
\end{enumerate}

\subsection{Extracting Compressed Keys}
Sort key compression can be done by extracting the bits in the positions which have value 1 in the D-bitmap. We now describe how to get compressed keys from index keys. 
(Though the examples in Figure~\ref{fig-compkey} are shown in the big endian format for readability, the actual implementation was done in the little endian format due to Intel processors.)

\begin{enumerate}
\item{Separate one-word long (8 bytes) masks from the D-bitmap. The first mask starts from the byte which contains the first 1 in the bitmap, and it is 8 bytes long. The second mask starts from the byte which contains the first 1 after the first mask, and it is 8 bytes long, and so on. In the example of Figure~\ref{fig-compkey}, we get three masks from the D-bitmap. See Figure~\ref{fig-compkey} (c).}

\item{By BMI instruction PEXT \cite{BMI} (which copies selected bits from the source to contiguous low-order bits of the destination), extract bits from an index key which are located in the positions where the masks have value 1. See Figure~\ref{fig-compkey} (d).}

\item{Concatenate the extracted bits with shift and bitwise OR operations. Since there are three masks in Figure~\ref{fig-compkey}, the extracted bits are concatenated in three steps (f).(i), (f).(ii), and (f).(iii) by a shift and a bitwise OR in each step. See Figure~\ref{fig-compkey} (e) and Figure~\ref{fig-compkey} (f).}
\end{enumerate}
The bit string in Figure~\ref{fig-compkey} (f).(iii) is the compressed key extracted from the full key in Figure~\ref{fig-compkey} (a).

\begin{figure}
\centering
\includegraphics[height=7cm]{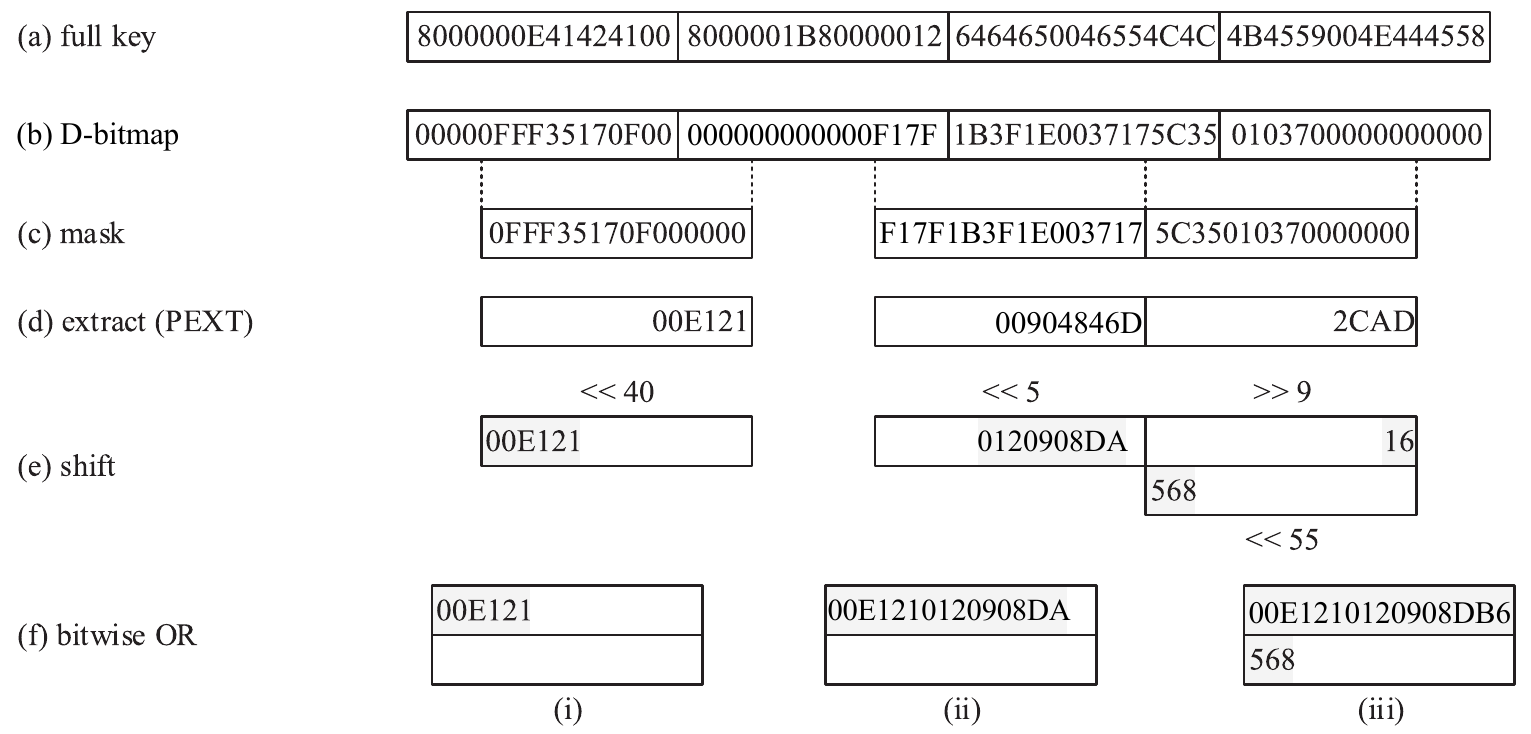}
\caption{Extracting compressed keys from index keys.}
\label{fig-compkey}
\end{figure}

\subsection{Parallel Sorting}
In our target applications which require a wide range of index key sizes, the size of sort keys is usually too big to exploit SIMD parallelism. 
Thus we implemented our own parallel sorting algorithm called the {\em row-column sort}, which is a \emph{comparison sort} \cite{CLRS} (i.e., it relies only on the operation of comparing two elements during sorting; the comparison operator is called a \emph{comparator}).
Hence the row-column sort works for any key sizes. 
The details of the row-column sort are described in Appendix.

\subsection{Parallel Index Construction}
Once the pairs of compressed index key and record ID are sorted in lexicographic order, the index tree can be built in a bottom-up fashion. First, we build leaf nodes from the sorted compressed keys and record IDs. To compute distinction bit positions, we make an array $\text{D-offset}[i]$ from the D-bitmap, which stores the position of the $(i+1)$-st 1 in the D-bitmap. Then the distinction bit position of $key_i$ and $key_{i+1}$ is $\text{D-offset}[\text{D-bit}(\text{Compress}(key_i), \text{Compress}(key_{i+1}))]$. Next, we build non-leaf nodes in a bottom-up fashion. For two adjacent entries in a non-leaf node whose highest keys are $key_i$ and $key_j$, the distinction bit position is $\text{D-offset}[\text{D-bit}(\text{Compress}(key_i), \text{Compress}(key_j))]$.

In the case of our index tree, the leaf nodes and non-leaf nodes contain partial keys of a predefined length $pk$. Given the offset (i.e., distinction bit position) of a partial key and the predefined partial key length $pk$, the bits of the partial key are determined as follows.

\begin{enumerate}
\item[A.]If a bit position of the partial key is included in the compressed key, the bit value can be directly copied from the compressed key. 

\item[B.]	If a bit position is a position which has value 0 in the variant bitmap (i.e., an invariant bit position), the bit value can be copied from the reference key value.

\item[C.]	Otherwise (i.e., a bit position which has value 0 in the D-bitmap and value 1 in the variant bitmap), we have two options.
\begin{enumerate}
\item[a)]	Add the bits required for partial key construction ($pk$ bits following the distinction bit position) to the compressed key and use them here for index construction.

\item[b)]	Since the record ID is also contained in the sort key, necessary bits can be copied from the record, for which a dereferencing is required.
\end{enumerate}
\end{enumerate}

To build an index, we maintain two parameters: max fanout and a fill factor. Each (leaf or non-leaf) node is of size 256B, and it has a header (24B). A leaf node also has a pointer to the next node (8B). Since each entry in a leaf node takes 16B, the max fanout (i.e., maximum number of entries) in a leaf node is $(256-32)/16=14$. Since each entry in a non-leaf node takes 24B, the max fanout in a non-leaf node is 9. The fill factor is defined for each index during index building, and leaf and non-leaf nodes are filled up to max fanout $\times$ fill factor \cite{CLRS}. Given the number of records, the max fanouts, and the fill factor (default value is 0.9), the height of the index tree can be determined.

Index construction can be parallelized by partitioning the sorted pairs of index key and record ID and constructing subtrees in parallel. That is, $n$ sort keys are divided into $p$ blocks of $\frac{n}{p}$ sort keys each, and one block is assigned to a thread (which is the situation at the end of the row-column sort). Thread $i$ $(1\leq i \leq p)$ constructs a subtree consisting of all sort keys in the $i$-th block. When all the subtrees are constructed, they are merged into one tree such that the height of the resulting tree can be minimized. Since the fanouts of the root nodes of the subtrees can be much less than the max fanout, just linking the root nodes of the subtrees may increase the height of the whole tree unnecessarily. Hence we remove the root nodes of the subtrees, and build the top layers of the whole tree by linking the children of the root nodes of the subtrees. In this way the height of the whole tree can be minimized.

\begin{table}
\renewcommand{\tabcolsep}{1.5mm}
\centering
\caption{Statistics of six datasets, where k = thousand, M = million, B = byte, and b = bit.}
\label{tab-stats}       
\begin{tabular}{ccccccc}
\hline
 & INDBTAB & Human& Wikititle & ExURL & WikiURL & Part\\
 
\hline
database table size & 884MB & 5310MB & 623MB & 649MB &930MB & 116MB \\
index size&390MB&860MB&333MB&184MB&305MB&46MB\\
\# keys & 16,392k & 36,504k & 13,978k & 7,735k & 12,786k & 2,000k \\
min key length & 35B & 101B & 2B &9B &31B&21B\\
max key length & 35B & 101B & 252B & 512B & 281B & 52B\\
average key length &35B& 101B & 21.2B & 59.0B &50.0B &33.7B\\ \hline
\# full key bits & 280b & 808b & 2016b & 4096b & 2248b & 416b \\
\# distinction bits in keys & 56b & 303b & 888b & 2023b & 874b & 204b\\
\# variant bits in record IDs & 27b & 29b & 26b & 25b & 26b & 25b\\
compression ratio &5.00&2.67&2.27&2.02&2.57&2.04\\
full sort key size (unit: 8B)& 48B & 112B & 264B & 520B & 296B & 64B\\
compressed sort key size (unit: 8B)& 16B & 48B & 120B & 256B & 120B & 32B\\
sort key ratio & 3.00 & 2.33 & 2.20 & 2.03 & 2.47 & 2.00\\
word comparison ratio & 2.10 & 1.27 & 1.00 & 1.37 & 3.61 & 1.29\\
\hline
\end{tabular}
\end{table}

\section{Performance Evaluation}

\subsection{Experimental Settings}

We conduct experiments to measure the performance improvements due to our compressed key sort.
In the experiments we compare the compressed key sort against the full key sort with respect to the time for sorting and index building.
We use five real datasets and one TPC-H dataset:
a database table in SAP HANA that records items in sales documents (which we call INDBTAB), a complete EST (expressed sequence tag) database of 
Human Chromosome 14 from Genome Assembly Gold-standard Evaluations \cite{HUMAN},
Wikipedia titles \cite{WIKITITLE}, external links of DBpedia \cite{DBpedia}, Wikipedia links of DBpedia \cite{DBpedia}, and Part table (column \emph{name}) of TPC-H \cite{TPC}.

The computer used in our experiments is equipped with four Intel\textsuperscript{\textregistered} Xeon\textsuperscript{\textregistered} E7-8880 v4 (2.20GHz) processors, each of which contains 22 cores. The computer has 1TB DRAM memory. (Since we used no more than 16 cores in our experiments of parallelization, the experiments were done in a single processor.)

\subsection{Evaluation with Real and Synthetic Datasets}
\label{sec-eval}


Table~\ref{tab-stats} presents the basic statistics of the six datasets such as the sizes of each database table and its index tree as well as a few important characteristics and measurements relevant to our proposed scheme. 
The {\em full sort key} refers to the combination of an uncompressed key taken from a dataset and the corresponding record ID. The record ID is 8 bytes long, and either the whole or only the variant bits of a record ID can be used as part of a sort key. In the latter case, the variant bitmap in the DS-metadata should be expanded to include the variant bits of the record IDs. 
The {\em compressed sort key} consists of distinction bits in a key and variant bits in the corresponding record ID.
Table~\ref{tab-stats} also shows the number of keys in a dataset, the lengths of the shortest, average, and longest keys, the number of bits in a full key, the number of distinction bits in keys, the number of variant bits in record IDs,
the size of full sort keys (i.e., full key + record ID), the size of compressed sort keys (i.e., distinction bits in key + variant bits in record ID).
The length of a sort key - full or compressed - is presented in the unit of 8B because sort keys are stored in {\em words}.

The compression ratio and the sort key ratio are computed by the following formulas:
\begin{eqnarray*}
\text{compression ratio} &=& \text{\# full key bits / \# distinction bits in keys} \\
\text{sort key ratio} &=& \text{full sort key size / compressed sort key size.}
\end{eqnarray*}
The compression ratio of our key compression scheme is 2.76 on average for the six datasets. (The percentage of distinction bit positions in full index keys is 39.8\% on average.)
The sort key ratio is 2.34 on average for the six datasets.
(The word comparison ratio will be explained in the next subsection.)
In all the experiments, comparison of two sort keys is made in the unit of 8-byte words.
Thus, for example, if two sort keys are 24B long, at most three word-comparisons will be required.

\begin{figure}
\centering
\includegraphics[height=5cm]{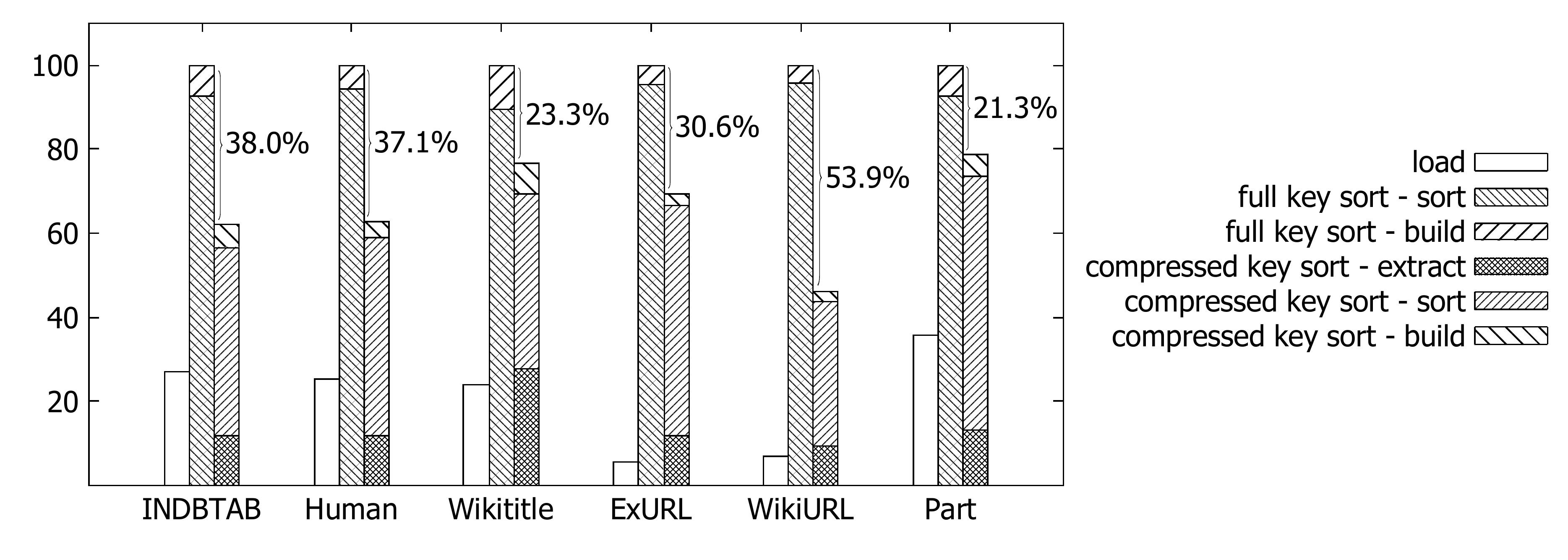}
\caption{Sorting time and index-building time of six datasets.}
\label{fig-bar}
\end{figure}

Figure~\ref{fig-bar} shows the performance results from all the six datasets. The execution times of the full key sort and the compressed key sort are summarized in three bars for each dataset in the figure.
The first bar (in the white color) of each group shows the time taken to load the key column of each database table from disk.
The load time is common to both sort methods and is included in the figure to show how the I/O cost is compared with the sort cost. The second and third bars of each group represent the total cost of building an index, excluding the load time, required by the full key sort and the compressed key sort, respectively. While the second bar is broken down to two phases, \emph{sort} and \emph{build}, the third bar is broken down to three phases, \emph{extract}, \emph{sort}, and \emph{build}. This is because the extract phase is needed only for the compressed key sort to obtain compressed keys by extracting bits from full keys in positions of 1s in the D-bitmap.
Despite the extra phase of bit extraction, however, as is shown clearly in Figure~\ref{fig-bar}, our compressed key sort reduced the total index building time substantially by expediting both the sort and build phases. The improvement ratio was 34.0\% on average for the six datasets. Note that all the measurements in the figure are normalized to the same scale for the ease of presentation and comparison.
The total time of building an index by the full key sort is 4.59, 21.12, 5.68, 17.63, 19.36, and 0.70 seconds for INDBTAB, Human, Wikititle, ExURL, WikiURL, and Part, respectively.

\begin{table}
\centering
\caption{Distribution of distinction bit positions of INDBTAB.}
\label{tab-dbit}
\begin{tabular}{c|c} 
\hline
bytes & distinction bit positions\\ 
\hline
1--8 &00000000 00000000 00000001 00000001 00000001 00001111 00000011 00001111 \\
9--16 &00000000 00000011 00001111 00000111 00001111 00000111 00001111 00000000 \\ 
17--24 &00001111 00001111 00001111 00001111 00001111 00001110 00000000 00000000 \\
25--32 &00000000 00000000 00000000 00000000 00000000 00000000 00000000 00000000 \\
33--35 &  \multicolumn{1}{l}{00000000 00000000 00000000}  \\ 
\hline
\end{tabular}
\end{table}

To better understand the performance differential, we looked into the distribution of distinction bit positions. Table~\ref{tab-dbit} shows the distinction bit positions of INDBTAB, where distinction bit positions are set by 1. As is shown in the table, the distinction bit positions are distributed widely over many bytes of full index keys, and extracting distinction bits into a compressed key can make it shorter and improve the performance of sorting and index building. In the case of full key sort, a single key comparison will have to examine up to 22 bytes of each key (i.e., sort by distinguishing prefixes), because the last distinction bit position is in the 22nd byte in the table. In the case of compressed key sort, however, a single key comparison can be done by examining no more than 7 bytes, because there are only 56 distinction bits, which can be stored in 7 bytes. Although a compressed sort key is actually 16 bytes long due to the 27 variant bits in the record IDs, a comparison of two compressed sort keys finishes in one word-comparison because all distinction bits belong to the first word of a compressed sort key.

\subsection{Sensitivity Analysis}

We conduct a sensitivity analysis to see how our sort key compression scheme performs under various circumstances.
The main parameters that affect the performance are the {\em sort key ratio} defined in the previous section and the \emph{word comparison ratio} defined as follows:
\[ \text{word comparison ratio} = \frac{wcc_{full}}{wcc_{comp}} \]
where $wcc_{full}$ is the average count of word comparisons required by a single full key comparison, and $wcc_{comp}$ is the average count of word comparisons required by a single compressed key comparison.



We used Zipf distribution~\cite{Zipf}
to generate synthetic datasets of various configurations.
Each dataset is generated by a custom function, denoted by
$\zipf(s,n,m)$, so that it contains 10 million keys of $n$ bytes each, the first $m$ bytes of each 8 byte word in a key have the same arbitrary ASCII value, and the remaining $8-m$ bytes of each word have lower case ASCII characters following the Zipf distribution $Zipf(s,26)$. The parameter $s$ is the value of the exponent characterizing the Zipf distribution.
For example, $\zipf(s, 16, 3)$ generates keys of type aaaZZZZZ aaaZZZZZ, where a is an arbitrary \textcolor{black}{fixed} character and Z is a byte having one of `a' to `z' by the Zipf distribution $(s,26)$.

\begin{table}
\renewcommand{\tabcolsep}{1.5mm}
\centering
\caption{Statistics of synthetic datasets.}
\label{tab-synthetic}       
\begin{tabular}{ccccccc}
\hline
 \multirow{2}{*}{data}&\multirow{2}{*}{function}&\multirow{2}{*}{key size}&full sort &compressed&sort key&word comparison  \\
&&&key size & sort key size&ratio&ratio \\
 
\hline
1&Zipf(2.5,48,0)&48B&56B&40B&1.40&1.30\\
2&Zipf(2.5,56,0)&56B&64B&40B&1.60&1.30\\
3&Zipf(2.5,64,0)&64B&72B&40B&1.80&1.30\\
4&Zipf(2.5,72,0)&72B&80B&40B&2.00&1.30\\
5&Zipf(2.5,80,0)&80B&88B&40B&2.20&1.30\\
6&Zipf(2.5,88,0)&88B&96B&40B&2.40&1.30\\
7&Zipf(2.5,96,0)&96B&104B&40B&2.60&1.30\\
8&Zipf(2.5,104,0)&104B&112B&40B&2.80&1.30\\
9&Zipf(2.5,112,0)&112B&120B&40B&3.00&1.30\\
\hline
10&Zipf(1.5,40,0)&40B&48B&24B&2.00&1.06\\
11&Zipf(1.5,40,1)&40B&48B&24B&2.00&1.11\\
12&Zipf(1.5,40,2)&40B&48B&24B&2.00&1.20\\
13&Zipf(1.5,40,3)&40B&48B&24B&2.00&1.34\\
14&Zipf(1.5,40,4)&40B&48B&24B&2.00&1.55\\
\hline
15&Zipf(1.5,64,0)&64B&72B&24B&3.00&1.05\\
16&Zipf(1.5,64,1)&64B&72B&24B&3.00&1.10\\
17&Zipf(1.5,64,2)&64B&72B&24B&3.00&1.19\\
18&Zipf(1.5,64,3)&64B&72B&24B&3.00&1.33\\
19&Zipf(1.5,64,4)&64B&72B&24B&3.00&1.53\\
20&Zipf(1.5,64,5)&64B&72B&24B&3.00&1.85\\
\hline
\end{tabular}
\end{table}

\begin{figure}
\centering
\subfigure{\includegraphics[height=5cm]{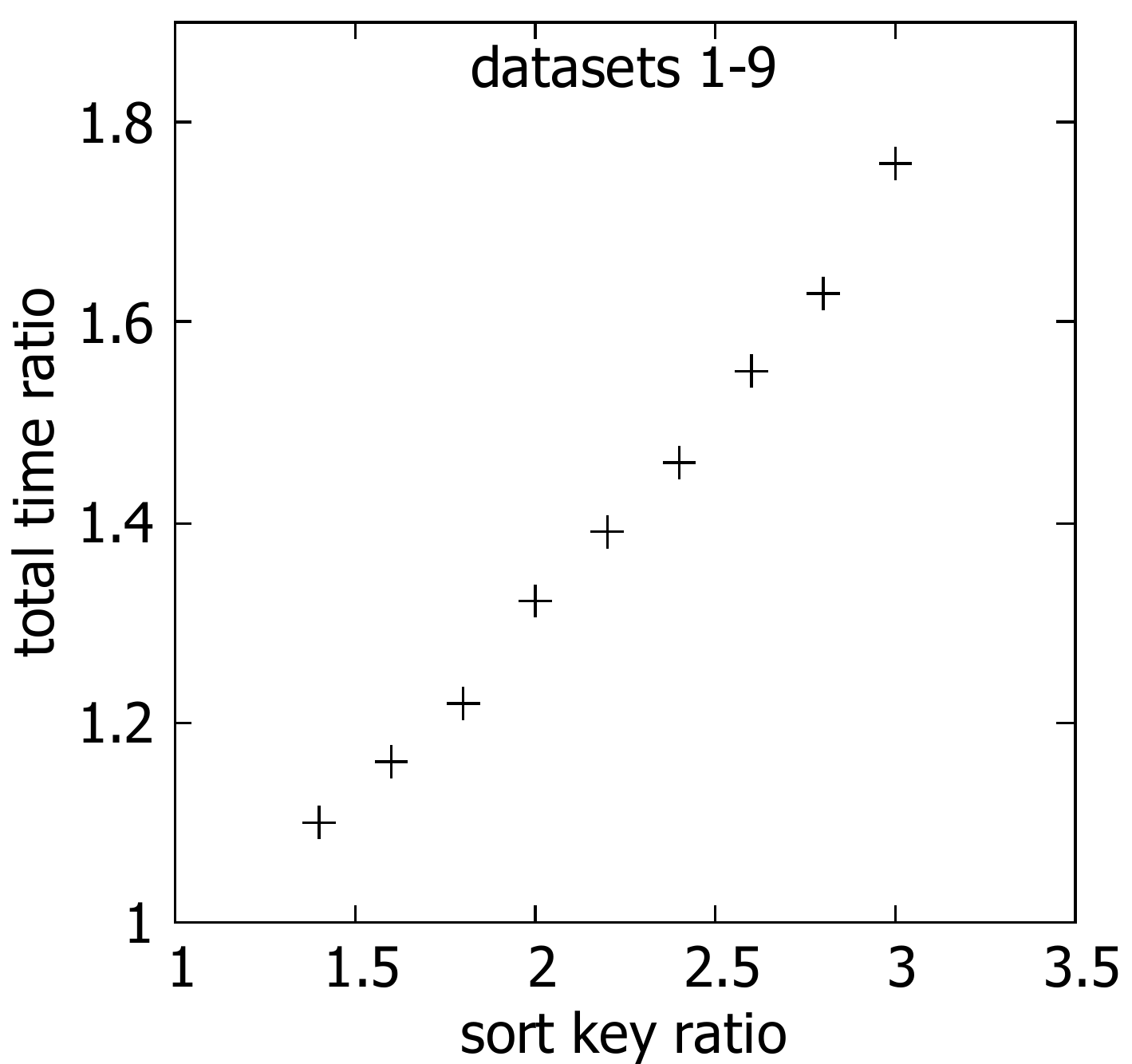}}
\subfigure{\includegraphics[height=5cm]{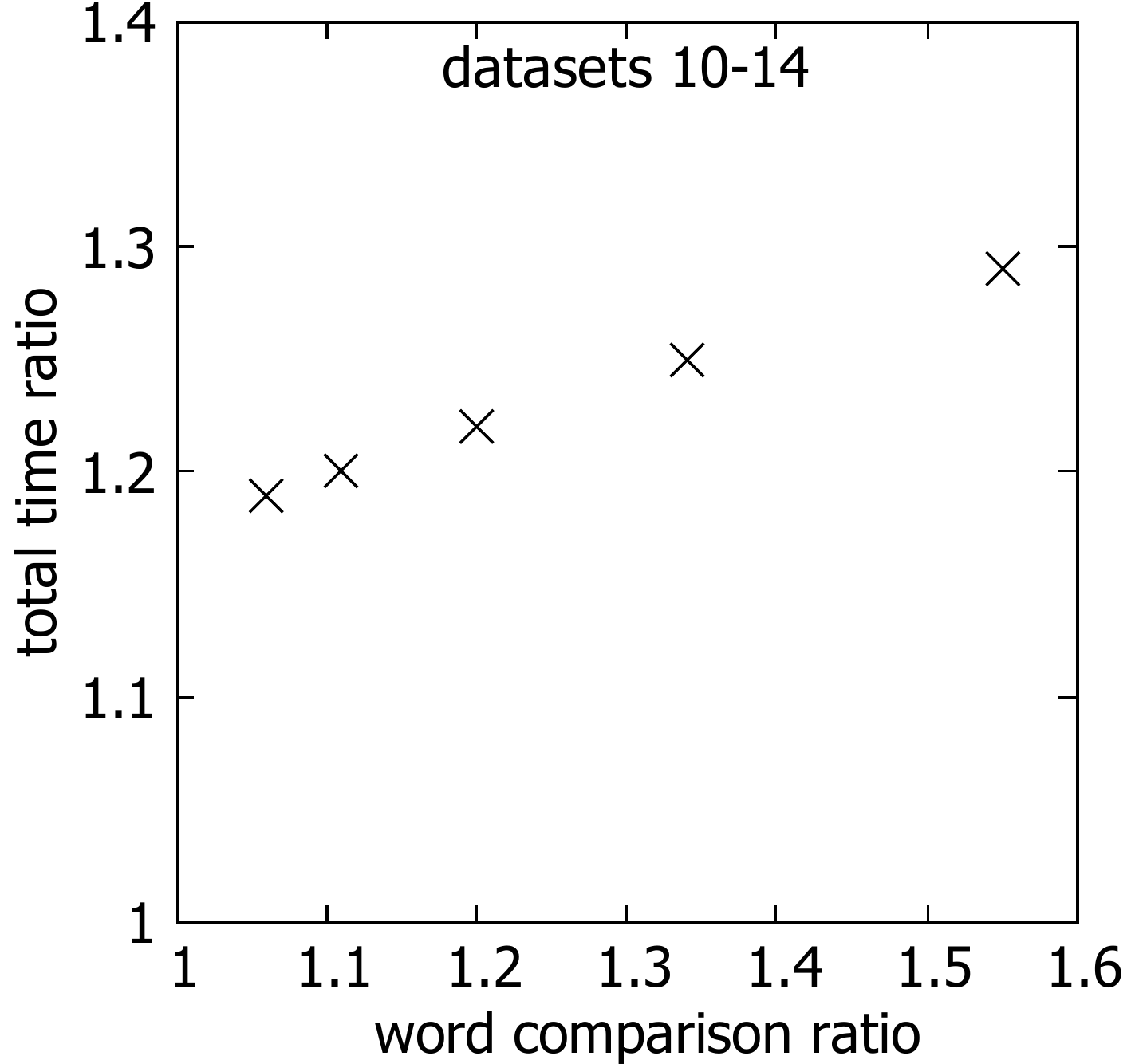}}
\subfigure{\includegraphics[height=5cm]{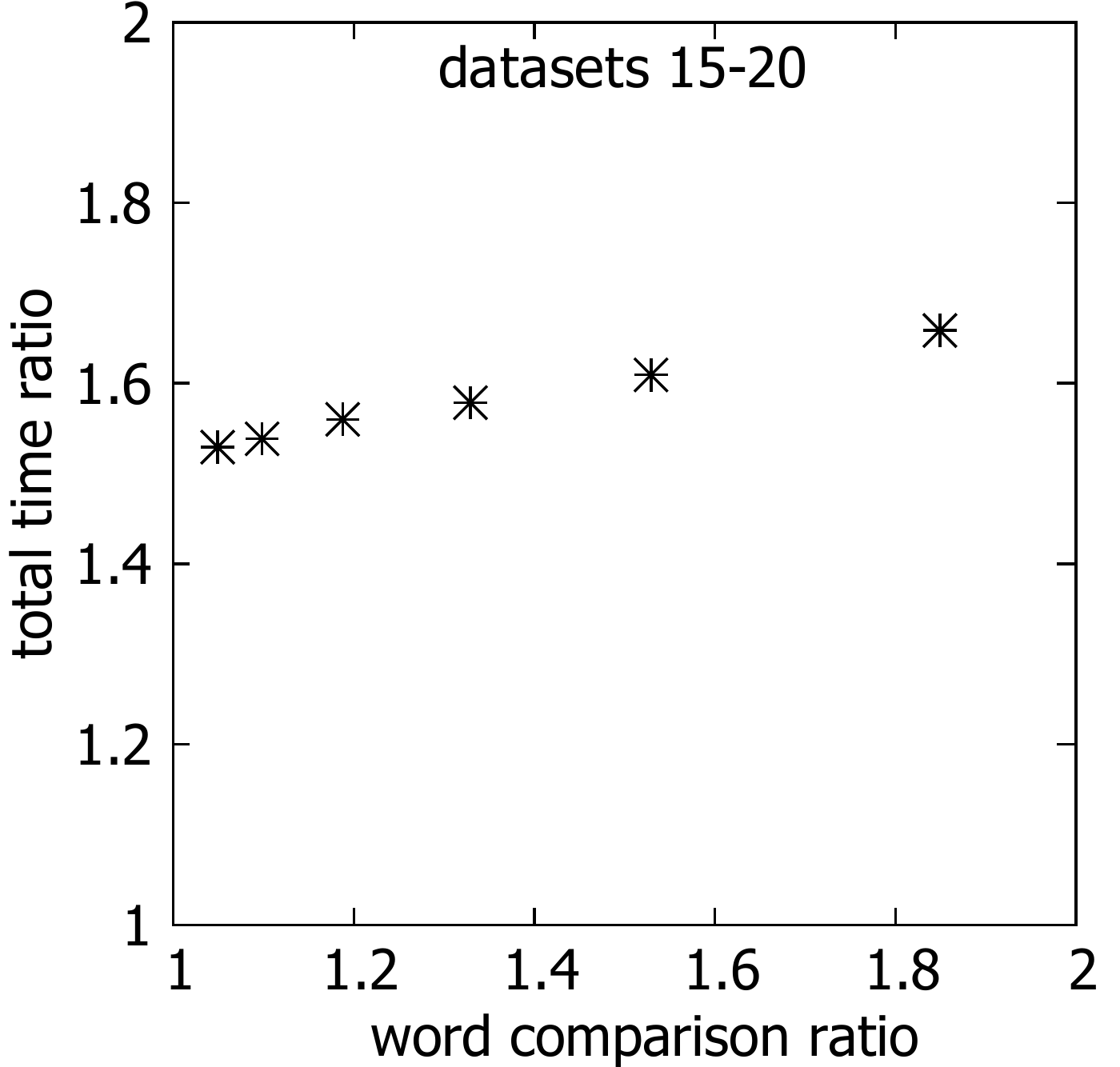}}
\caption{Total time ratio of synthetic datasets.}
\label{fig-synthetic}
\end{figure}

Table~\ref{tab-synthetic} shows the statistics of the synthetic datasets used in the sensitivity analysis. As can be seen in the table, the datasets are generated such that they conform to the default parameter settings: word comparison ratio = 1.30 and sort key ratio = 2.00 or 3.00.
In datasets 1-9, the word comparison ratio is fixed to 1.30 and the sort key ratio changes as the key size changes. In datasets 10-14 (resp. 15-20), the sort key ratio is fixed to 2.00 (resp. 3.00) and the word comparison ratio changes as $m$ changes in $\zipf(s,n,m)$. 

Figure~\ref{fig-synthetic} shows the index construction times by the compressed key sort in comparison with those by the full key sort. The measurements are the ratio of the former to the latter. So, the higher the ratio is, the larger the margin of improvement is by the compressed key sort. In datasets 1-9, as the sort key ratio 
increases, the advantage of our compressed key sort scheme grows proportionally. This is because the size of compressed sort keys gets smaller, which leads to a smaller amount of work in sorting and index building.

In datasets 10-14 (also in datasets 15-20) the number of distinction bits in each dataset is about the same, and thus the sort key ratios are identical. However, as $m$ in $\zipf(1.5, 40, m)$ increases, the distinction bits are more widely spread in full keys. Hence, our compression scheme has the effect of compacting widely spread distinction bits in full keys into compressed keys, which leads to a smaller number of word-comparisons in a comparison of two compressed keys. As $m$ increases, therefore, the word comparison ratio gets bigger, which results in an improvement especially in sorting time, even though the sort key ratios 
remain the same.
Therefore, our compression scheme improves the performance of index building in two ways:
\begin{enumerate}
\item by making compressed keys shorter than full keys, which leads to a smaller amount of work in sorting and index building, and
\item by compacting distinction bits in full keys into compressed keys, which leads to a smaller number of word-comparisons in a comparison of two compressed keys.
\end{enumerate}
For example, whereas the sort key ratio of WikiURL is smaller than that of INDBTAB (in Table~\ref{tab-stats}), the improvement of WikiURL is larger than that of INDBTAB (in Figure~\ref{fig-bar}) because the word comparison ratio of WikiURL is larger than that of INDBTAB.

Finally, we compare the sensitivity analysis to the experimental results with the six datasets with respect to (sort key ratio, word comparison ratio, total time ratio). 
INDBTAB has a characteristic (3.00, 2.10, 1.61), which is similar to (3.00, 1.85, 1.66) of dataset 20; (2.33, 1.27, 1.59) of Human is similar to (2.40, 1.30, 1.46) of dataset 6; (2.20, 1.00, 1.30) of Wikititle is similar to (2.00, 1.06, 1.19) of dataset 10; (2.03, 1.37, 1.44) of ExURL is similar to (2.00, 1.34, 1.25) of dataset 13.

\begin{table}
\centering
\caption{Sorting time and index-building time of INDBTAB (in seconds).}
\label{tab-indbtab}       
\begin{tabular}{c|ccc|cccc|cc|cc}
\hline
 & \multicolumn{3}{c|}{full key sort} & \multicolumn{4}{c|}{compressed key sort} & \multicolumn{2}{c|}{total time} & \multicolumn{2}{c}{speedup}\\
cores & sort & build & total & extract & sort & build & total & ratio& improve & full & comp\\
\hline
1 & 4.251 & 0.340 & 4.591 & 0.543 & 2.063 & 0.242 & 2.848 & 1.61 & 38.0\% & 1.0 & 1.0 \\
2 & 2.195 & 0.171 & 2.366 & 0.274 & 1.053 & 0.123 & 1.450 & 1.63 & 38.7\% & 1.9 & 2.0 \\
4 & 1.181 & 0.090 & 1.271 & 0.138 & 0.572 & 0.066 & 0.776 & 1.64 & 38.9\% & 3.6 & 3.7 \\
8 & 0.549 & 0.049 & 0.598 & 0.069 & 0.265 & 0.035 & 0.369 & 1.62 & 38.3\% & 7.7 & 7.7 \\
16 & 0.310 & 0.034 & 0.344 & 0.034 & 0.148 & 0.025 & 0.207 & 1.66 & 39.8\% & 13.3 & 13.8 \\

\hline
\end{tabular}
\end{table}

\begin{figure}
\centering
\includegraphics[height=5.0cm]{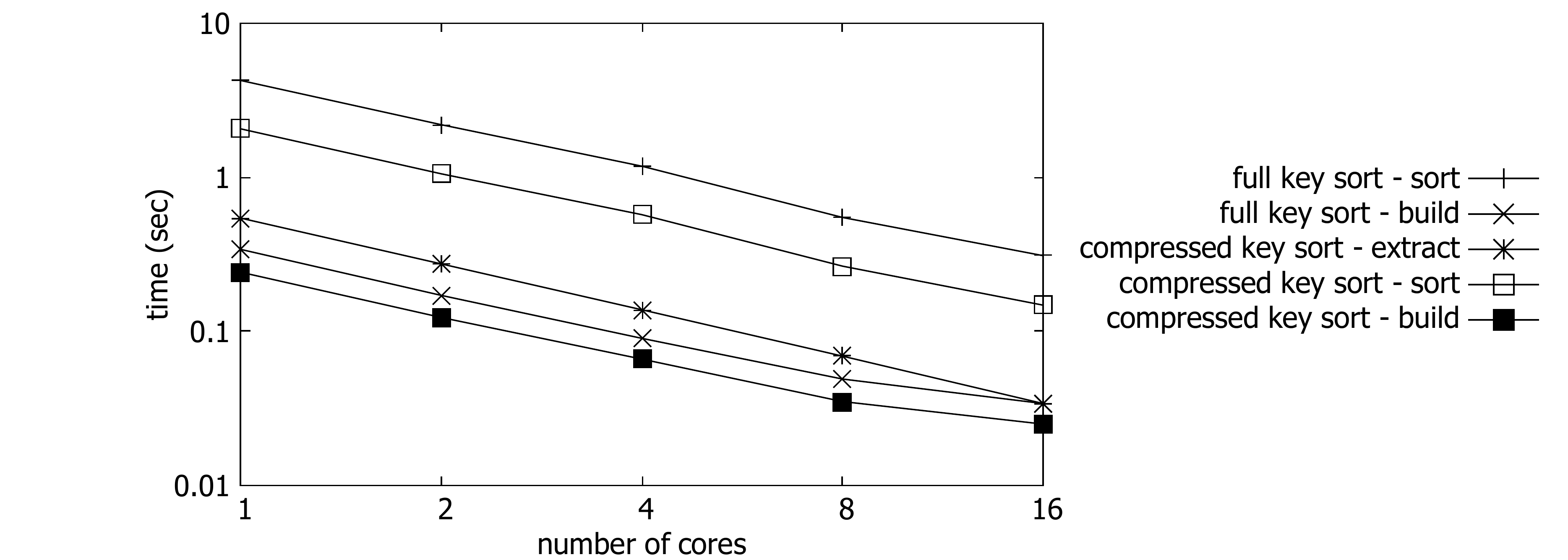}
\caption{Speedups of sorting time and index-building time.}  
\label{fig-indbtab}
\end{figure}

\subsection{Parallelization and Loading}

In Section~\ref{sec-eval}, we have shown that our compressed key sort can reduce times for both sort and index-build phases considerably. Nonetheless, the sort time still remains as the dominant portion of total index reconstruction time. In this section, we show that the sort time could be further reduced by parallelizing it on a multi-core computing platform. The choice of our sort algorithm was the {\em row-column sort}. Refer to~\ref{app-sort} for the detailed description of the algorithm.

Table~\ref{tab-indbtab} shows the detailed performance measurements from the full key sort and the compressed key sort with a varying number of cores used. Figure~\ref{fig-indbtab} presents the speedups observed in this experiment in the log-log scale. It clearly shows near-linear speedup in all measurements from both the full key sort and compressed key sort except for the index-building time with 16 cores, in which the speedup deteriorated because memory write (390MB in Table~\ref{tab-stats}) became dominant and was not accelerated by using multiple cores.
When sixteen cores were used, the speedups in the total index reconstruction time were 13.3 and 13.8 for the full key sort and the compressed key sort, respectively.

In the case of sixteen cores, the tree index of INDBTAB was reconstructed from the pre-loaded database table in just 0.207 seconds. Given that the memory footprint of the index tree of INDBTAB is 390MB (see Table~\ref{tab-stats}), this is approximately equivalent to 1.88 GB/sec bandwidth. This level of bandwidth is higher than the peak bandwidth of enterprise class magnetic disk drives and most commodity solid state drives. We present this result as an evidence that a tree index can be reconstructed from the pre-loaded database table on the fly much more efficiently than loading the materialized image from disk.



\section{Conclusion}

We have defined the notion of distinction bit positions and proved that the bit slices of index keys in distinction bit positions are sufficient information to determine the sorted order of the index keys correctly. Consequently, utilizing only those bit slices is in effect equivalent to compressing keys losslessly in regard to sorting the keys. The key compression ratio achieved by the proposed method was 2.76:1 on average in our experiment.

We have then proposed the compressed key sort based on the distinction bit positions in order to expedite the reconstruction of a tree index from the base table in memory. The compressed key sort reduced the index reconstruction time by 34.0\% on average in our experiment carried out on a single core platform. 
The proposed method based on distinction bit positions is essentially a lightweight key compression scheme. Thus it can be adopted in any application that involves sorting keys longer than the word size and is expected to deliver significant performance benefit.

\section*{Acknowledgements}

Ryu and Park were supported by Institute for Information \& communications Technology Promotion(IITP) grant funded by the Korea government(MSIT) (No. 2018-0-00551, Framework of Practical Algorithms for NP-hard Graph Problems). Moon was supported by a grant (K-16-L03-C01-S03) funded by the ministry of science, ICT, and future planning, Korea and PF Class Heterogeneous High Performance Computer Development (NRF-2016M3C4A7952587).

\bibliography{ref}   

\begin{thebibliography}{10}
\expandafter\ifx\csname url\endcsname\relax
  \def\url#1{\texttt{#1}}\fi
\expandafter\ifx\csname urlprefix\endcsname\relax\def\urlprefix{URL }\fi
\expandafter\ifx\csname href\endcsname\relax
  \def\href#1#2{#2} \def\path#1{#1}\fi

\bibitem{Zhang16hstore}
H.~Zhang, D.~G. Andersen, A.~Pavlo, M.~Kaminsky, L.~Ma, R.~Shen, {Reducing the
  Storage Overhead of Main-Memory OLTP Databases with Hybrid Indexes}, in:
  Proceedings of the 2016 ACM SIGMOD International Conference on Management of
  Data, ACM, 2016, pp. 1567--1581.

\bibitem{sqlserver}
{Microsoft, SQL Server Documentation},
  \url{https://docs.microsoft.com/en-us/sql/relatio}\break \url{
  nal-databases/in-memory-oltp/comparing-disk-based-table-storage-to-memory-opt}\break
  \url{imized-table-storage?view=sql-server-2017}.

\bibitem{hekaton}
C.~Diaconu, C.~Freedman, E.~Ismert, P.-A. Larson, P.~Mittal, R.~Stonecipher,
  N.~Verma, M.~Zwilling, Hekaton: Sql server's memory-optimized oltp engine,
  in: Proceedings of the 2013 ACM SIGMOD International Conference on Management
  of Data, ACM, 2013, pp. 1243--1254.

\bibitem{MMDB}
F.~Faerber, A.~Kemper, P.-A. Larson, J.~Levandoski, T.~Neumann, A.~Pavlo, Main
  memory database systems, Foundations and Trends in Databases 8~(1-2) (2017)
  1--130.

\bibitem{malviya14}
N.~Malviya, A.~Weisberg, S.~Madden, M.~Stonebraker, {Rethinking Main Memory
  OLTP Recovery}, in: Proceedings of the 30th International Conference on Data
  Engineering, IEEE Computer Society, 2014, pp. 604--615.

\bibitem{BU}
R.~Bayer, K.~Unterauer, Prefix {B}-trees, ACM Transactions on Database Systems
  2~(1) (1977) 11--26.

\bibitem{LC}
T.~J. Lehman, M.~J. Carey, A study of index structures for main memory database
  management systems, in: Proceedings of the 12th International Conference on
  Very Large Data Bases, Morgan Kaufmann Publishers Inc., 1986, pp. 294--303.

\bibitem{Fe}
D.~E. Ferguson, Bit-tree: A data structure for fast file processing,
  Communications of the ACM 35~(6) (1992) 114--120.

\bibitem{BMR}
P.~Bohannon, P.~Mcllroy, R.~Rastogi, Main-memory index structures with
  fixed-size partial keys, ACM SIGMOD Record 30~(2) (2001) 163--174.

\bibitem{RR99}
J.~Rao, K.~A. Ross, Cache conscious indexing for decision-support in main
  memory, in: Proceedings of the 25th International Conference on Very Large
  Data Bases, Morgan Kaufmann Publishers Inc., 1999, pp. 78--89.

\bibitem{RR00}
J.~Rao, K.~A. Ross, Making {B}+- trees cache conscious in main memory, ACM
  SIGMOD Record 29~(2) (2000) 475--486.

\bibitem{Chen01}
S.~Chen, P.~B. Gibbons, T.~C. Mowry, Improving index performance through
  prefetching, ACM SIGMOD Record 30~(2) (2001) 235--246.

\bibitem{Chen02}
S.~Chen, P.~B. Gibbons, T.~C. Mowry, G.~Valentin, Fractal prefetching
  {B}+-trees: Optimizing both cache and disk performance, in: Proceedings of
  the 2002 ACM SIGMOD International Conference on Management of Data, ACM,
  2002, pp. 157--168.

\bibitem{KARY}
B.~Schlegel, R.~Gemulla, W.~Lehner, K-ary search on modern processors, in:
  Proceedings of the 15th International Workshop on Data Management on New
  Hardware, ACM, 2009, pp. 52--60.

\bibitem{BSV}
M.~Boehm, B.~Schlegel, P.~B. Volk, U.~Fischer, D.~Habich, W.~Lehner, Efficient
  in-memory indexing with generalized prefix trees, in: Proceedings of the 14th
  BTW conference on Database Systems for Business, Technology, and Web, 2011,
  pp. 227--246.

\bibitem{KISS}
T.~Kissinger, B.~Schlegel, D.~Habich, W.~Lehner, {KISS}-tree: Smart latch-free
  in-memory indexing on modern architectures, in: Proceedings of the 18th
  International Workshop on Data Management on New Hardware, ACM, 2012, pp.
  16--23.

\bibitem{KCS}
C.~Kim, J.~Chhugani, N.~Satish, E.~Sedlar, A.~D. Nguyen, T.~Kaldewey, V.~W.
  Lee, S.~A. Brandt, P.~Dubey, {FAST}: Fast architecture sensitive tree search
  on modern {CPU}s and {GPU}s, in: Proceedings of the 2010 ACM SIGMOD
  International Conference on Management of Data, ACM, 2010, pp. 339--350.

\bibitem{YOH}
T.~Yamamuro, M.~Onizuka, T.~Hitaka, M.~Yamamuro, {VAST}-tree: A vector-advanced
  and compressed structure for massive data tree traversal, in: Proceedings of
  the 15th International Conference on Extending Database Technology, ACM,
  2012, pp. 396--407.

\bibitem{bwtree}
J.~J. {Levandoski}, D.~B. {Lomet}, S.~{Sengupta}, The {Bw}-tree: A {B}-tree for
  new hardware platforms, in: Proceedings of the 29th International Conference
  on Data Engineering, IEEE Computer Society, 2013, pp. 302--313.

\bibitem{LKN}
V.~Leis, A.~Kemper, T.~Neumann, The adaptive radix tree: {ART}ful indexing for
  main-memory databases, in: Proceedings of the 29th International Conference
  on Data Engineering, IEEE Computer Society, 2013, pp. 38--49.

\bibitem{SuRF}
H.~Zhang, H.~Lim, V.~Leis, D.~Andersen, M.~Kaminsky, K.~Keeton, A.~Pavlo,
  {SuRF}: practical range query filtering with fast succinct tries, in:
  Proceedings of the 2018 ACM SIGMOD International Conference on Management of
  Data, ACM, 2018, pp. 323--336.

\bibitem{HOT}
R.~Binna, E.~Zangerle, M.~Pichl, G.~Specht, V.~Leis, {HOT}: A height optimized
  trie index for main-memory database systems, in: Proceedings of the 2018 ACM
  SIGMOD International Conference on Management of Data, ACM, 2018, pp.
  521--534.

\bibitem{GL}
G.~Graefe, P.-A. Larson, B-tree indexes and {CPU} caches, in: Proceedings of
  the 17th International Conference on Data Engineering, IEEE Computer Society,
  2001, pp. 349--358.

\bibitem{G11}
G.~Graefe, Modern {B}-tree techniques, Foundations and Trends in Databases
  3~(4) (2011) 203--402.

\bibitem{ORC}
{Apache ORC}, \url{https://orc.apache.org/}.

\bibitem{PAR}
{Apache Parquet}, \url{https://parquet.apache.org/}.

\bibitem{OP}
G.~Antoshenkov, D.~Lomet, J.~Murray, Order preserving string compression, in:
  Proceedings of the 12th International Conference on Data Engineering, IEEE
  Computer Society, 1996, pp. 655--663.

\bibitem{DOP}
G.~Antoshenkov, Dictionary-based order-preserving string compression, The VLDB
  Journal 6~(1) (1997) 26--39.

\bibitem{QUERY}
Z.~Chen, J.~Gehrke, F.~Korn, Query optimization in compressed database systems,
  ACM SIGMOD Record 30~(2) (2001) 271--282.

\bibitem{DOPC}
C.~Binnig, S.~Hildenbrand, F.~F\"{a}rber, Dictionary-based order-preserving
  string compression for main memory column stores, in: Proceedings of the 2009
  ACM SIGMOD International Conference on Management of Data, ACM, 2009, pp.
  283--296.

\bibitem{OPKC}
H.~Zhang, X.~Liu, D.~Andersen, M.~Kaminsky, K.~Keeton, A.~Pavlo,
  Order-preserving key compression for in-memory search trees, in: Proceedings
  of the 2020 ACM SIGMOD International Conference on Management of Data, ACM,
  2020, pp. 1601--1615.

\bibitem{IMK}
H.~Inoue, T.~Moriyama, H.~Komatsu, T.~Nakatani, {AA}-sort: A new parallel
  sorting algorithm for multi-core {SIMD} processors, in: Proceedings of the
  16th International Conference on Parallel Architecture and Compilation
  Techniques, IEEE Computer Society, 2007, pp. 189--198.

\bibitem{CMB}
J.~Chhugani, A.~D. Nguyen, V.~W. Lee, W.~Macy, M.~Hagog, Y.-K. Chen,
  A.~Baransi, S.~Kumar, P.~Dubey, Efficient implementation of sorting on
  multi-core {SIMD CPU} architecture, Proceedings of the VLDB Endowment 1~(2)
  (2008) 1313--1324.

\bibitem{SKC}
N.~Satish, C.~Kim, J.~Chhugani, A.~D. Nguyen, V.~W. Lee, D.~Kim, P.~Dubey, Fast
  sort on {CPU}s and {GPU}s: A case for bandwidth oblivious {SIMD} sort, in:
  Proceedings of the 2010 ACM SIGMOD International Conference on Management of
  Data, ACM, 2010, pp. 351--362.

\bibitem{GCC}
{The GNU C++ library manual},
  \url{http://gcc.gnu.org/onlinedocs/libstdc++/manual/}.

\bibitem{Patent}
Y.~S. Kwon, K.~Park, C.~Yoo, Optimal sort key compression and index rebuilding,
  {US Patent Application Number 15/658,671} (2017).

\bibitem{SKS}
A.~Silberschatz, H.~F. Korth, S.~Sudarshan, Database Systems Concepts, 4th
  Edition, McGraw-Hill Higher Education, 2001.

\bibitem{FAST}
V.~Leis, {FAST} source, \url{http://www-db.in.tum.de/~leis/index/fast.cpp}.

\bibitem{CSB}
J.~Rao, {CSB}+ tree source,
  \url{http://www.cs.columbia.edu/~kar/software/csb+}.

\bibitem{BMI}
Intel, {A}dvanced Vector Extensions Programming Reference, 2011.

\bibitem{CLRS}
T.~H. Cormen, C.~E. Leiserson, R.~L. Rivest, C.~Stein, Introduction to
  Algorithms, 3rd Edition, The MIT Press, 2009.

\bibitem{HUMAN}
{Genome datasets of Human Chromosome 14},
  \url{http://gage.cbcb.umd.edu/data/index.html}.

\bibitem{WIKITITLE}
{Wikipedia titles dump}, \url{http://dumps.wikimedia.org/enwiki/}.

\bibitem{DBpedia}
S.~Auer, C.~Bizer, G.~Kobilarov, J.~Lehmann, R.~Cyganiak, Z.~Ives, {DB}pedia: A
  nucleus for a web of open data, in: The Semantic Web, Springer Berlin
  Heidelberg, 2007, pp. 722--735.

\bibitem{TPC}
M.~Poess, C.~Floyd, New {TPC} benchmarks for decision support and web commerce,
  ACM SIGMOD Record 29~(4) (2000) 64--71.

\bibitem{Zipf}
S.~Ross, A First Course in Probability, 6th Edition, Prentice Hall, 2002.

\bibitem{VSI}
P.~J. Varman, S.~D. Scheufler, B.~R. Iyer, G.~R. Ricard, Merging multiple lists
  on hierarchical-memory multiprocessors, Journal of Parallel and Distributed
  Computing 12~(2) (1991) 171--177.

\bibitem{FMP}
R.~S. Francis, I.~D. Mathieson, L.~Pannan, A fast, simple algorithm to balance
  a parallel multiway merge, in: Proceedings of the 5th International PARLE
  Conference on Parallel Architectures and Languages Europe, Springer-Verlag,
  1993, pp. 570--581.

\bibitem{BM}
J.~L. Bentley, M.~D. McIlroy, Engineering a sort function, Software: Practice
  and Experience 23~(11) (1993) 1249--1265.

\bibitem{LL}
A.~LaMarca, R.~E. Ladner, The influence of caches on the performance of
  sorting, Journal of Algorithms 31~(1) (1999) 66--104.

\bibitem{Kn}
D.~E. Knuth, The Art of Computer Programming, Volume 3: Sorting and Searching,
  2nd Edition, Addison Wesley Longman Publishing Co., Inc., 1998.

\bibitem{SS}
H.~Shi, J.~Schaeffer, Parallel sorting by regular sampling, Journal of Parallel
  and Distributed Computing 14~(4) (1992) 361--372.

\end{thebibliography}

%
%

\appendix

\renewcommand{\thefigure}{\arabic{figure}}
\renewcommand{\thetable}{\arabic{table}}
\section{Parallel Sorting}
\label{app-sort}

\begin{algorithm}[h]
\caption{Row-Column Sort}
\begin{algorithmic}[1]
\Procedure{Row\_Column\_Sort}{Key[1..$n$], $n$, $p$, $e$, $C$}
    \State $c \gets \lfloor C/e \rfloor$
    \State $t \gets \max(\lfloor \frac{\sqrt{n/c}}{p}\rfloor,1)$ 
    \For{$i \gets 1$ \textbf{to} $tp$}  
            \State init\_block$[i] \gets (\frac{n}{tp}(i-1)+1, \frac{n}{tp}i)$
            \State sorted\_block$[i] \gets (\frac{n}{tp}(i-1)+1, \frac{n}{tp}i)$ \Comment{in another array Temp}
            \For{$j \gets 1$ \textbf{to} $\frac{n}{tpc}$} 
                \State sub\_block$[i][j] \gets (\frac{n}{tp}(i-1) + c(j-1) + 1,\frac{n}{tp}(i-1) + cj)$
            \EndFor
    \EndFor
    
    \ForAll{thread $i \gets 1$ \textbf{to} $p$} \textbf{in parallel}
        \For{$j \gets 1$ \textbf{to} $t$}
            \For{$k \gets 1$ \textbf{to} $\frac{n}{tpc}$}
                \State basic\_sort(sub\_block[$t(i-1)+j$][$k$])
            \EndFor
            \State multiway\_merge($\frac{n}{tpc}$, sub\_block[$t(i-1)+j$][$1..\frac{n}{tpc}$], sorted\_block[$t(i-1)+j$]) 
        \EndFor
    \EndFor
    \ForAll{thread $i \gets 1$ \textbf{to} $p$} \textbf{in parallel}
        \State perfect\_partition(sorted\_block$[1..tp]$, split\_block$[1..tp][1..p]$)
    \EndFor
    \ForAll{thread $i \gets 1$ \textbf{to} $p$} \textbf{in parallel}
        \State final\_block$[i] \gets (\frac{n}{p}(i-1)+1,\frac{n}{p}i)$ 
        \State multiway\_merge($tp$, split\_block$[1..tp][i]$, final\_block[$i$]) 
    \EndFor

\EndProcedure
\end{algorithmic}
\end{algorithm}


We describe our parallel sorting algorithm, which we call the row-column sort. The row-column sort uses a notion of the \emph{perfect partition} in \cite{VSI,FMP}. A pair of a (full or compressed) key and its record ID will be called a sort key, which is an element in sorting. 
The input to the row-column sort is as follows:

\begin{enumerate}
\item[] Key$[1..n]$: array of sort keys
\item[] $n$: number of elements (i.e., sort keys)
\item[] $p$: number of threads
\item[] $e$: size of an element (in bytes)
\item[] $C$: last level cache size (in bytes) per thread (i.e., available L3 cache size / number of threads in our experiments)
\end{enumerate}

For the dataset of INDBTAB in Section 6, for instance, $n$ is 16.39 million, and $e$ is 48 bytes for full sort keys. Typically in our target applications, the size of sort keys is too big to exploit SIMD parallelism. Hence, the row-column sort does not rely on SIMD instructions, but it is a \emph{comparison sort} \cite{CLRS} (i.e., it relies on the operation of comparing two elements).
Algorithm 1 shows the pseudo-code of the row-column sort.
The details of the algorithm are as follows.

\begin{figure}
\centering
\includegraphics[height=7cm]{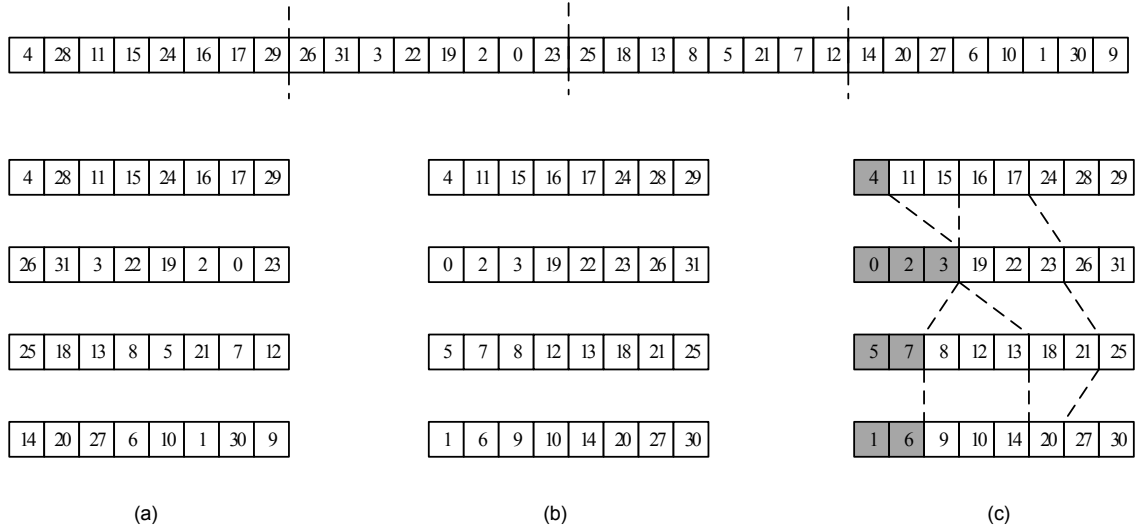}
\caption{Row-column sort. (a) init\_block$[1..tp]$, where $t=1$, $p=4$. (b) sorted\_block$[1..tp]$ (each block is sorted). (c) split\_block$[1..tp][1..p]$ (perfect split of sorted blocks).}
\label{fig-sort}
\end{figure}

\begin{enumerate}
\item[1.] (lines 2--3) Compute two parameters which are used in the algorithm: $c$ is the number of elements that can be included in $C$ bytes, and $t$ is set such that $\frac{n}{tpc} \approx tp$ in order to balance the workloads of line 13 and line 18.
In Figure~\ref{fig-sort}, $n=32$ and $p=4$. For simplicity, we assume in this toy example that $e=1$ and $c=2$. Then $t$ is set to 1.

\item[2.] The row-column sort uses two arrays Key$[1..n]$ and Temp$[1..n]$ which are partitioned into blocks: init\_block$[1..tp]$, sub\_block$[1..tp][1..\frac{n}{tpc}]$, and final\_block$[1..p]$ are blocks of array Key, and sorted\_block$[1..tp]$ and split\_block$[1..tp][1..p]$ are blocks of array Temp. See Figure~\ref{fig-sort}.
For simplicity of presentation, we assume that $\frac{n}{tp}$, $\frac{n}{tpc}$, and $\frac{n}{p}$ are integers.
In Algorithm 1, each block is represented by the first position and the last position in its array (but in actual implementation only one of the first and last positions is necessary because the whole array is partitioned into blocks without overlaps).
For example, init\_block[1] is represented by $(1, \frac{n}{tp})$.

\item[3.] (lines 9--13) Assign $t$ init\_blocks to each thread. Each thread sorts each of $t$ init\_blocks as follows. (Note that a block init\_block$[i]$ is partitioned into sub\_block$[i][1..\frac{n}{tpc}]$.)
\begin{enumerate}
\item[3.1.] 
Sort each sub\_block$[i][k]$ of Key$[1..n]$ by the following basic sort. The basic sort is essentially Quicksort with insertion sort as the recursion base. The Quicksort partitions around the median of the medians of three samples, each of three elements (also called the pseudo-median of 9 elements) \cite{BM}. This basic sort is fast when all the input elements are within the last level cache \cite{BM,LL}.

\item[3.2.] Each thread merges $\frac{n}{tpc}$ sub\_blocks into a sorted\_block 
by the multi-way merge (i.e., $\frac{n}{tpc}$-way merge) using a tournament tree \cite{Kn}.
(In multiway\_merge($x$, in\_block$[1..x]$, out\_block) of Algorithm 1, $x$ is the number of blocks to be merged, in\_block$[1..x]$ are the blocks to be merged, and out\_block is the merged block.)
\end{enumerate}

\item[4.] (lines 14--15) Compute the perfect $p$-partition of the $tp$ sorted\_blocks \cite{VSI,FMP}.
The \emph{perfect $p$-partition} of sorted\_blocks is defined as follows: Each sorted\_block is partitioned into $p$ split\_blocks (sizes of split\_blocks may vary and there can be even an empty split\_block as in the second sorted\_block of Figure~\ref{fig-sort} (c)) such that the collection of all the first split\_blocks constitutes the $\frac{n}{p}$ smallest ones of $n$ elements (gray elements in Figure~\ref{fig-sort} (c)), and the collection of all the second split\_blocks constitutes the next $\frac{n}{p}$ smallest ones, etc.

\item[5.] (lines 16--18) Thread $i$ ($1\leq i \leq p$) merges all the $i$-th split\_blocks of the perfect $p$-partition (i.e., split\_block$[1..tp][i]$) into a final\_block. Again we use the multi-way merge (i.e., $tp$-way merge) using a tournament tree.
\end{enumerate}

\begin{table}
\centering
\caption{GCC STL sort vs. row-column sort (in seconds).}
\label{tab-psort}
\begin{tabular}{c|ccccc} 
\hline
cores & 1 & 2 & 4 & 8 & 16\\ 
\hline
GCC STL sort & 4.016 & 2.727 & 1.380 & 0.775 & 0.452  \\
row-column sort & 4.251 & 2.195 & 1.181 & 0.549 & 0.310  \\
\hline
\end{tabular}

\end{table}

\begin{figure}
\centering
\includegraphics[height=5cm]{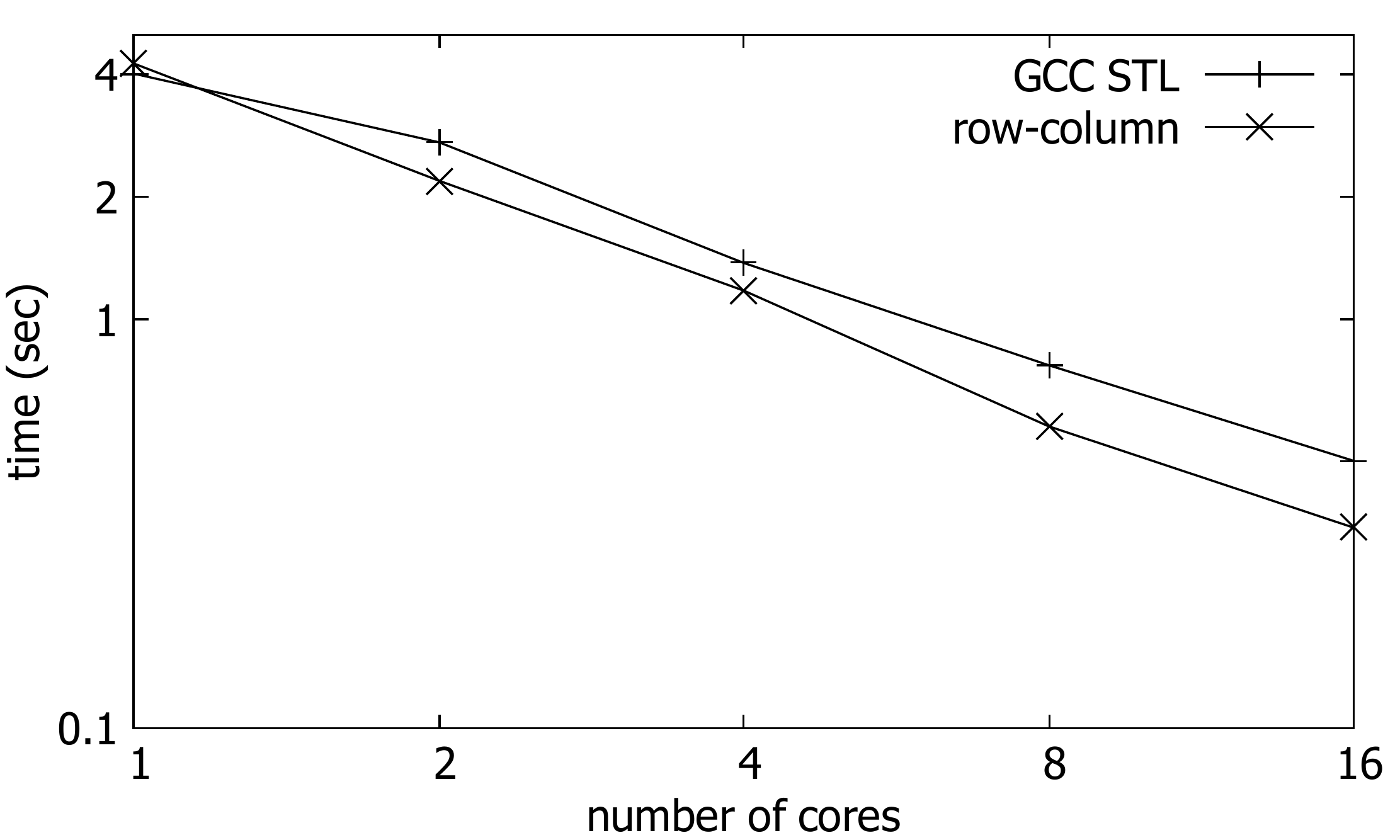}
\caption{Speedups of GCC STL sort and row-column sort.}
\label{fig-psort}
\end{figure}

The perfect $p$-partition in step 4 is computed as follows: An \emph{$x$-split} of the $tp$ sorted\_blocks is defined as a partition of each sorted\_block$[i]$ $(1\leq i \leq tp)$ into two disjoint subsets $L_i$ and $H_i$ such that
\begin{enumerate}
\item[(1)] any element in all $L_i$'s is less than or equal to any element in all $H_i$'s 
\item[(2)] the number of elements in all $L_i$'s is exactly $x$. 
\end{enumerate}
To find the perfect $p$-partition, each thread $i$ ($1\leq i \leq p-1$) computes an $i\times \frac{n}{p}$-split of the $tp$ sorted\_blocks. 
Then the $\frac{n}{p}\text{-split}, \frac{2n}{p}\text{-split}, \ldots, \frac{(p-1)n}{p}\text{-split}$
make the perfect $p$-partition. 
In Figure~\ref{fig-sort} (c), the $\frac{n}{p}\text{-split}$ has $L_1=\{4\}$, $L_2=\{0,2,3\}$, $L_3=\{5,7\}$, $L_4=\{1,6\}$, and we set split\_block$[i][1]=(\frac{n}{tp}(i-1)+1,\frac{n}{tp}(i-1)+|L_i|)$, i.e., split\_block$[1][1]=(1,1)$, split\_block$[2][1]=(9,11)$, split\_block$[3][1]=(17,18)$, split\_block$[4][1]=(25,26)$.
We use the algorithm due to Francis, Mathieson and Pannan \cite{FMP} to find an $x$-split. The algorithm in \cite{VSI} also computes an $x$-split, and one in \cite{SS} computes an approximate split.


For the performance of the row-column sort, we compared it with GCC STL parallel sort \cite{GCC}, which is also a comparison sort in which a custom comparator can be used.
Table~\ref{tab-psort} shows the sorting times of GCC STL sort and the row-column sort for full sort keys of INDBTAB, and Figure~\ref{fig-psort} shows the speedups of the two sorting algorithms with multiple cores.
The row-column sort shows a better speedup than GCC STL sort, and
it is 31.4\% faster than GCC STL sort when the number of cores is 16.

\end{document}